\documentclass{llncs}
\usepackage{amssymb}
\usepackage{amsfonts}
\usepackage{amsmath}
\usepackage{fancyhdr}
\usepackage{epsfig}
\usepackage{mathtools}
\usepackage{graphicx}
\usepackage {longtable}
\usepackage{lineno,hyperref}
\usepackage{graphicx}
\usepackage{tabularx}
\usepackage{enumitem}
\usepackage{algorithm}
\usepackage{algorithmic}
\usepackage{wrapfig}
\usepackage{longtable}
\usepackage{xcolor}
\usepackage{tikz}
\usetikzlibrary{positioning}
\usepackage{standalone}
\usepackage{subcaption}
\usepackage{url}
\usepackage{booktabs,multirow}
\newcommand{\Ic}{\ensuremath{\mathcal  I}}
\newcommand{\Jc}{\ensuremath{\mathcal  J}}

\newcommand{\rb}{\ensuremath{\mathbf{r}}}
\newcommand{\bfs}{\ensuremath{\mathbf{s}}}

\newcommand{\xb}{\ensuremath{\mathbf{x}}}

\newcommand{\beg}{\begin{example}}
\newcommand{\eeg}{\end{example}}
\newcommand{\bit}{\begin{itemize}}
\newcommand{\eit}{\end{itemize}}
\newcommand{\beq}{\begin{equation}}
\newcommand{\eeq}{\end{equation}}
\newcommand{\beqn}{\begin{equation*}}
\newcommand{\eeqn}{\end{equation*}}
\newcommand{\bea}{\begin{eqnarray}}
\newcommand{\eea}{\end{eqnarray}}
\newcommand{\bean}{\begin{eqnarray*}}
\newcommand{\eean}{\end{eqnarray*}}
\newcommand{\ben}{\begin{enumerate}}
\newcommand{\een}{\end{enumerate}}
\newcommand{\bprop}{\begin{proposition}}
\newcommand{\eprop}{\end{proposition}}
\newcommand{\bthm}{\begin{theorem}}
\newcommand{\ethm}{\end{theorem}}
\newcommand{\blem}{\begin{lemma}}
\newcommand{\elem}{\end{lemma}}
\newcommand{\bcor}{\begin{corollary}}
\newcommand{\ecor}{\end{corollary}}
\newcommand{\bdefn}{\begin{definition}}
\newcommand{\edefn}{\end{definition}}
\newcommand{\brmk}{\begin{remark}}
\newcommand{\ermk}{\end{remark}}

\newcommand{\slen}[1]{\setlength\itemsep{#1 mm}}

\begin{document}

\title{Equivalent Non-linear Shift Register Transformations and Espresso}
\titlerunning{NFSR Transformations of Espresso}
\author{Anirban Ghatak \inst{1}\orcidID{0000-0002-8654-2411} \and Kaushik Sarkar \inst{2} \and Anupam Chattopadhyay \inst{3} \and Ambrish Awasthi \inst{4}\orcidID{0009-0000-9066-5708} \and Indivar Gupta\inst{4}}
\authorrunning{A. Ghatak et al.}
\institute{Indian Statistical Institute, Kolkata, India\\
\email{ghatak.anirban@gmail.com}
\and Narula Institute of Technology, Kolkata\\
\email{kaushik.srkr@gmail.com}
 \and CCDS, Nanyang Technological University, Singapore\\
 \email{anupam@ntu.edu.sg}
\and Scientific Analysis Group, DRDO, Metcalfe House, Delhi, India, 110054\\
\email{ambrishawasthi.sag@gov.in}}
\maketitle
\begin{abstract}
	The stream cipher Espresso was designed by E. Dubrova and M. Hell as a hardware-efficient candidate for 5G security. It has two components - a $ 256 $-bit binary non-linear feedback shift register (NFSR), with $14$ feedback taps, and a $ 20 $-variable non-linear output function. However, the security of Espresso with regard to several standard cryptanalytic attacks are estimated using an alternate NFSR, which has only two feedback taps. The above analysis is based on the claim that this alternate NFSR is obtained using a transformation algorithm which preserves the output sequence of the original NFSR. This issue has become all the more significant in the light of recent results on hardware performance evaluations, which assume that the stated equivalence holds.\\ 
	This article presents the first successful attempt in settling the above question of Espresso NFSR equivalence. We arrive at our main result by a careful analysis of existing transformation methods relevant to the question. In doing so, we first identify and correct a critical flaw in the proof of a foundational result of E. Dubrova. Next, in the context of Espresso equivalence, we analyze and unify the transformation framework of Dubrova and the one proposed by Yao and Parampalli. Based on our analysis, we have succeeded in transforming the LFSR in an Espresso-equivalent design reported by Yao and Parampalli to the Espresso alternate NFSR. Further, we have obtained a corresponding non-linear output function which achieves output-equivalence with the LFSR design, and hence, with Espresso.
\end{abstract}
\section{Introduction}
Linear and non-linear feedback shift registers (LFSRs and NFSRs/NLFSRs) are key design components of several families of stream ciphers. The properties of LFSRs have been well-studied since the seminal work of S. Golomb \cite{Golomboss}, and a recent work \cite{claudpalashi} rigorously showcases both new analyses/construction and a survey of key contributions pertaining to LFSR-based designs.  In comparison, the analysis of NFSRs is still a work in progress. But recent advances indicate that it may be possible to provide the same level of cryptographic security as LFSR-based designs using NFSRs but with a noticeably reduced hardware footprint. Efforts in this direction have resulted in the design of NFSR-based stream ciphers like LIZARD \cite{lizard}, TRIVIUM \cite{trivium}, Grain \cite{sasya}  and, more recently, Espresso \cite{dubrova17}. However, providing security guarantees for NFSR-based stream ciphers is more of a challenge than LFSR-based designs owing to the comparative difficulty in NFSR analysis. A possible way out may be to formulate an \emph{equivalent design} involving an LFSR, or a simpler NFSR, whose analysis is more tractable.\\
Another aspect of significance in the design of shift-register based stream ciphers is the choice between \emph{Galois} and \emph{Fibonacci} configurations of shift registers. In a Fibonacci NFSR there is only a single feedback update function from the intermediate stages to the terminal stage, all the other stages accept a single input from just the preceding stage. Hence, the Fibonacci configuration is a restricted case of the general Galois configuration, which potentially allows feedback between any pair of stages in the shift register. In \cite{Mansodubri}, it is claimed that a stream cipher based on a Galois-configuration shift register can achieve higher clock frequency on application-specific integrated circuit (ASIC). Again, the security analysis of a stream cipher using the Fibonacci configuration is comparatively easier and well-studied in literature.\\
The stream cipher Espresso was proposed by E. Dubrova and M. Hell \cite{dubrova17} as a hardware-efficient candidate for providing security in 5G communication. This cipher offers an excellent case study for the aforementioned competing design choices - Galois vs. Fibonacci and LFSR vs. NFSR. As a contender for 5G security, Espresso claims better performance than Grain-128 and Trivium in the under-$1500$ Gate Equivalence (GE) chip-size regime with the stated parameters: $ 1497 $ GE area, $ 2.22 $ Gbits/sec throughput and $ 232 $ nanosecond latency.\\
Espresso is based on a $ 256 $-bit Galois NFSR, having $ 14 $ non-linear feedback update functions at different stages, and a non-linear output function, involving the outputs of $ 20 $ different stages. However, the security of Espresso against standard cryptanalysis strategies like linear approximation attacks and TMDTO attacks is actually discussed in \cite{dubrova17} with respect to an ``equivalent'' NFSR (cf. p. 277 of \cite{dubrova17}), having only two non-linear update functions. Hence, Espresso is clearly an instance of the ``analysis-on-equivalent-design'' strategy mentioned earlier.\\
The motivation for the present work arose from an examination of the claimed equivalence of the two NFSRs presented in the Espresso proposal paper in the context of cryptanalysis of Espresso. An immediate issue had already been reported in \cite{yapa21}: no ``equivalent'' non-linear output function is provided in \cite{dubrova17} corresponding to the alternative Galois NFSR. This precludes any comparison between the actual keystreams generated by Espresso and an ``equivalent'' cipher based on the alternative NFSR. Further unsubstantiated claims made in \cite{dubrova17} in this context were reported in \cite{space25}.

A series of papers by Yao and Parampalli (\cite{yapa21}, \cite{yaoimproved}, \cite{yapa22}) culminated in defining the so-termed `compensation method' for transformations of several types of Fibonacci and Galois NFSRs. Their method applied to Espresso is claimed to yield (cf. \cite{yapa21}, Section 6.1):\\
``... a Fibonacci LFSR with a nonlinear output function. Specifically, the output function consists of $ 2289 $ monomials. There are $ 104 $ variables in the output function and the algebraic degree is $ 12 $.''\\
But the question of whether or not the two NFSRs presented in \cite{dubrova17} are actually equivalent was still not resolved. Notwithstanding the existence of this gap, some recent works have already attempted improved implementation of Espresso configurations (e.g. \cite{espeschindes}), assuming the validity of the equivalence as claimed in \cite{dubrova17}.\\

In this article, we expand on the results reported in \cite{space25} and attempt to conclude a proposed resolution to the \emph{Espresso equivalence problem} reported therein. We begin with a careful analysis of a sequence of transformation techniques which are outlined in E. Dubrova's works \cite{dubrova09}, \cite{dubrova10},\cite{dubrova14}, and those of G. Yao and U. Parampalli \cite{yapa21}, \cite{yapa22}) in the context of the equivalence problem. We then proceed to define a unified framework of the Dubrova and Yao-Parampalli algorithms, the \emph{Dubrova-Yao-Parampalli (D-Y-P) framework}, to gain new insights as to the claimed equivalence. Finally, we present experimental results and theoretical analysis to \emph{achieve this claimed equivalence} using the D-Y-P framework.\\
Specifically, our contributions, collated with those reported in \cite{space25}, are as follows.
\begin{enumerate}
	\item We identify a critical mis-step in the proof of the foundational theorem of Dubrova's first algorithm (cf. Theorem 1, Section IV, \cite{dubrova09}) and provide a correct proof (Section \ref{secbhuldubri}).
	\item We show that Dubrova's first algorithm \cite{dubrova09}, is essentially subsumed into the framework of the second algorithm \cite{dubrova14} (cf. \cite{space25}) (Subsection \ref{seceqdubri}).
	\item We establish that the Yao-Parampalli algorithm (\cite{yapa21}, \cite{yapa22}) is equivalent to Dubrova's second algorithm in the case of Fibonacci-to-uniform Galois NFSR transformation (cf. \cite{space25}) for computing the feedback functions and matching initial states (Subsection \ref{secyapdubri}).
	\item Using the D-Y-P framework, we obtain a NFSR/non-linear function ensemble from the Yao-Parampalli LFSR-based design, featuring the Espresso alternate NFSR $ F $ (cf. Subsection \ref{subs:lofornofor}). Specifically, we have obtained a non-linear output function corresponding to the NFSR $F$, thereby closing the gap in the equivalence claims made in \cite{dubrova17}.
	\item We have supported our experimental results\footnote{The Sagemath code for our reference implementation is available at:
		
		https://github.com/kaushiksrkr/YPFibonacci2EspressoFGalois.git} using the obtained `cipher' featuring the Espresso alternate NFSR $ F $, with theoretical analysis (cf. Section \ref{secbom}). Our results allow us to conclusively state that:\\
	\emph{Using the D-Y-P transformation framework, we have positively settled the Espresso Equivalence Problem by proving the equivalence of the Espresso designs in the sense of matched outputs of the two NFSR/non-linear function ensembles.}
\end{enumerate} 

%AnuComments:
%- tikk Page 3. Contributions. Could highlight (with hyper-ref) which contribution is located in which section of this paper, for easy navigation.
%- tikk Page 4: Reference implementation, either to be mentioned in footnote or in appendix.
%- Natikk Section 2.1 first para. “Under the following heads”. It could be some sub-sub-section directly. The triangle figure would be very apt at this position, highlighting which transformation is achieving which objective (security/implementation).
%- tikk Section 2 is like a preliminary note on what has been done by D, Y and P. It would be useful to put a figure or table on what transformation is proposed by whom in which paper, where theoretical/empirical validation is missing etc.. That sets the background on what’s to come in section 3 and 4. Otherwise, there is a lot of back-n-forth required to remember the connections. This is tried with the ‘observations’ and ‘remarks’ but, could be more systematically summarised. 
%- Can the theorem 5 statement be made w/o referring to theorem 1 within the text of it?
%- Page 25 appears more of the criticism of Y-P paper rather than presenting new result. Can we reduce or remove redundancy and rather have the logical flow like - this was claimed, this does not work, this is how we did it. In the current detailing, this part reads more like a ‘review of Y-P paper’.
%- tikk Towards the end of the paper, D-Y-P is mentioned together. It has to be stated somewhere, which theorems put together forms this framework.
\section{Preliminaries and Prior Work}
A foundational analysis of non-linear feedback shift registers, in particular, the cycle structure of the \emph{state transition graph}, is given, like all matters ``shift register'', in Golomb's text (\cite{Golomboss}, Part Three, Chapter VI onwards). We proceed to outline two major algorithms in the works of E. Dubrova, dealing with the transformations between Fibonacci and the so-called uniform Galois NFSRs. This is followed by a brief account of the NFSR transformation framework in the works of G. Yao and U. Parampalli.
\subsection{The Transformation Framework of E. Dubrova}\label{shegoldubri}
Our treatment of the transformation framework of E. Dubrova is based on the results in \cite{dubrova09}, \cite{dubrova10}, and \cite{dubrova14}. In later sections we show that there is an evolutionary thread underlying these works which culminated in the Espresso design in \cite{dubrova17}. In the following brief exposition, we segregate Dubrova's results in the three representative papers under the following heads.
\begin{itemize}
	\item Transformation from Fibonacci to \emph{uniform} Galois NFSR;
	\item Matching initial states in Fibonacci-uniform Galois NFSR transformations;
	\item Extension to (uniform) Galois-Galois NFSR transformations.
\end{itemize}
\subsubsection{Fibonacci to Uniform Galois NFSR Transformation \cite{dubrova09}}
The first algorithm of E. Dubrova in \cite{dubrova09} considers the transformation of a Fibonacci NFSR to a \emph{uniform} Galois NFSR, which guarantees the following notion of equivalence between the NFSRs (cf. \emph{Definition 1}, p.5265, \cite{dubrova09}):\\
``Two NLFSRs are equivalent if their sets of output sequences are equal.''\\
In \cite{dubrova09}, the \emph{output bit} of an NFSR is specified to be the $ 0 $-th bit (cf. Section II), and so, the transformation attempts to match the output sequence at the $ 0 $-th stage of the original Fibonacci NFSR with that at the $ 0 $-th stage of the transformed Galois NFSR.\\
The validity of this transformation is based on the uniform property of a Galois NFSR, which is defined as follows(cf. \emph{Definition 7}, \cite{dubrova09}).
\bdefn
The \emph{terminal bit} of an $ n $-bit NFSR is the bit with the minimal index which satisfies the following condition: For all bits $ i $ such that $ i<\tau $, the feedback function $ f_i $ is of type $ f_i = x_{(i+1)_n} $, where $(i+1)_n := (i+1) \mod n$. 
\edefn
\begin{definition}\label{defuni}
	An $ n $-bit NFSR is \emph{uniform} if 
	\begin{enumerate}[label = {(\alph*)}]
		\item all its feedback functions are singular functions of the following form: 
		\beq\label{eqfnuni}
		f_i (x_0, x_1, \cdots, x_{n-1}) = x_{(i+1)_n}\oplus g_i(x_0, x_1, \cdots, x_{n-1}) 
		\eeq
		$ g_i: \{0,1\}^{n-1} \rightarrow \{0,1\}  $ be such that $ (i+1)_n \notin \text{dep}(g_i) $.
		\item for all its bits $ i $ such that $ i>\tau $, the following condition holds:
		\beq\label{eqcondunif}
		\text{max\_index}(g_i)\leq \tau
		\eeq
		where $ \tau $ is the terminal bit of the NFSR.
	\end{enumerate}
\end{definition}
In the above context, the \emph{dependence set} of a Boolean function $ f $ is denoted by:
\[\text{dep}(f) := \{i \mid f\lvert_{x_i =0} \neq f\lvert_{x_i =1}\} \]
where $ f\lvert_{x_i =j} = f(x_0, x_1, \cdots, x_{i-1},j, x_{i+1}, \cdots, x_{n-1}), $ for $j\in \{0,1\}$.\\
The main algorithm presented in \cite{dubrova09} (cf. Section V) is a transformation algorithm from a Fibonacci to a \emph{fully shifted}, uniform Galois NFSR, which is claimed to preserve the output sequence at the $ 0 $-th stage. The term ``fully shifted" refers to the fact that this construction guarantees the maximum possible shift from the Fibonacci configuration while maintaining the uniform property.
\subsubsection*{Matching Initial States in Fibonacci-Galois Transformations \cite{dubrova10}}
In this work dealing with transformations between equivalent NFSRs, in the sense of \cite{dubrova09}, a formula for the \emph{matching} initial state of a uniform Galois NFSR is given in terms of the initial state of the Fibonacci NFSR.
\bthm[\emph{Theorem 3}, p. 2965, \cite{dubrova10}]\label{thm:compendubri}
Let $N_F$ be an $n$-bit Fibonacci NLFSR and $N_G$ be an equivalent to $N_F$ uniform Galois NLFSR with the terminal bit $\tau, \, 0\leq \tau < n-1$, and the feedback functions of type
\begin{align*}
	f_{n-1}(x) & = x_0 \oplus g_{n-1}(x_0, \ldots, x_{\tau})\\
	f_{n-1}(x) & = x_0 \oplus g_{n-1}(x_0, \ldots, x_{\tau})\\
	& \cdots\\
	f_{\tau}(x) & = x_{\tau +1} \oplus g_{n-1}(x_0, \ldots, x_{\tau})\\
	f_{\tau-1}(x) & = x_{\tau}\\
	& \cdots \\
	f_{0}(x) & = x_1	
\end{align*}
where $g_{n-1}, g_{n-2},\ldots, g_{\tau}$ are Boolean functions.\\
If $N_F$ is initialized to a state $s= (s_0,s_1, \ldots, s_{n-1})$, and $N_G$ is initialized to the state $r = (s_0, s_1, \ldots, s_{\tau}, r_{\tau +1}, r_{\tau +2}, \ldots, r_{n-1})$ such that
\beq\label{eqdubricompenis}
r_i = s_i \oplus g_{i-1}(s)\oplus g_{i-2}\lvert_{+1}(s)\oplus \cdots \oplus g_{\tau}\lvert_{+(i - \tau -1)}(s)
\eeq
for all $i \in \{n-1, n-2, \ldots, \tau +1\}$, then $N_F $ and $N_G$ generate the same output sequences.
\ethm
In Section \ref{secyapdubri}, it will be shown that this result can be interpreted in terms of \emph{simple compensation} in the Yao-Parampalli transformation framework \cite{yapa22} (cf. Subsection \ref{subs:yap}).
\subsubsection*{General Framework for Galois-to-Galois NFSR Transformations \cite{dubrova14}}
In \cite{dubrova14}, a more general algorithm was presented which handles transformations between two NFSRs in the Galois configuration via the operation of \emph{valid shifting} described as follows.
\begin{definition}\label{defshft}
	Given an $ n $-variate mapping defined by functions of the form in Equation \ref{eqfnuni}, a \emph{shifting} $ f_i\xrightarrow{m} f_j  $, $ i,j \in \{0,1,\ldots, n-1\}, \, i\neq j $, shifts the monomial $ m $ in the algebraic normal form (ANF) of $ f_i $ to the ANF of $ f_j $, such that every index $ a\in \text{dep}(m) $ is changed to $ b = (a-i+j) \mod n $.
\end{definition}
\begin{definition}\label{defvs}
	Consider an $ n $-variate mapping with functions of the form in Equation \ref{eqfnuni}, such that the values computed by $ f_{z}, \, z \in \{0,1,\ldots, n-1\} $ form the \emph{output sequence}. Then a shifting $ g_i\xrightarrow{m} g_j  $ is \emph{valid} if, for every index $ a\in \text{dep}(m) $ and $ b $ as in Definition \ref{defshft}, the following conditions hold.
	\begin{enumerate}
		\item For each $ c\in [a,b]\setminus{i} $, $ g_c =0 $; if $ i \in [a,b] $, then $ g_i^{\ast} = 0$ $ (g_i^{\ast} = g_i \oplus m) $;
		\item For each $ k \in [i,j] $, $ k\notin \text{dep}(g_i^{\ast}) $.\\
		Further, $ k \notin \text{dep}(g_p) $ for all $ p \in \{0,1,\ldots, n-1\}\setminus{i} $;
		\item Neither $ [a,b] $ nor $ [i,j] $ contain both $ z $ and $ z-1 $.
	\end{enumerate}
\end{definition}
In the above definition, $ [i,j] $ stands for the set $ \{i,i-1, \ldots, j\} $ if $ i>j $, and the set $ \{i,i+1, \ldots, j\} $ if $ i<j $, with similar definitions for $ [a,b] $. Further, if multiple indices $ z $ constitute the output sequence, the third condition must be fulfilled for all pairs $ (z,z-1) $.\\
The following result from \cite{dubrova14} establishes that when a Galois NFSR is transformed to another through valid shiftings, the output sequences of the two NFSRs may differ only in a specific subset of the stages of the internal state. 
\begin{theorem}[E. Dubrova \cite{dubrova14}]\label{thdubri2}
	Let $ F  $ be a mapping as defined by Equation \ref{eqfnuni} and $ F' $ be a mapping obtained from $ F $ by applying a valid shifting \( g_i \xrightarrow{\text{m}} g_j
	\), where \( i,j \in \{0,\ldots,n-1\} \) and \( i \neq j \). Denote $ \Jc =\{i+1, i+2, \ldots, j\} $ and $ \Ic = \{i, i-1, \ldots, j+1\} $. Let \( F \) be initialized to the state \( (s_0, s_1, \ldots, s_{n-1}) \) and \( F' \) to the state \( (r_0, r_1, \ldots, r_{n-1}) \) such that:
	\begin{enumerate}
		\item For $k \in  \Ic $, $r_k = s_k \oplus m |_{ k - i - 1}$
		
		\item For $ k \in  \Jc$, $ r_k = s_k \oplus m | _{k - j - 1}$
		
		\item \( r_k = s_k \) for all \( k \in \{0, 1, \ldots, n-1\} \) not in either $ \Ic $ or $ \Jc $.
	\end{enumerate}
	Then the sequences of states generated by \( F \) and \( F' \) may differ only in the bit positions in $ \Ic $ or $ \Jc $, respectively.
\end{theorem}
\brmk
The above theorem is a generalization of a result in \cite{dubrova10} (\emph{Theorem 2}, p. 2964), which was proved for the shift of a feedback function from the terminal bit $\tau$ to a new terminal bit $\tau -1$.
\ermk
\subsection{The Yao-Parampalli Framework}\label{subs:yap}
In \cite{yapa22}, a classification of Galois NFSRs based on the forms of the feedback functions is discussed, where four types of NFSRs are described and four sets of transformation algorithms are presented. Of interest in this article is the Type-II Galois NFSR, which corresponds to Dubrova's uniform Galois NFSR. Henceforth, we will refer to this configuration as uniform Galois.\\
The so-called \emph{compensation method} was defined to satisfy the notion of output equivalence where the output is a function of several stages of the shift register. In this context, we state the following definitions from \cite{yapa22}.
\bdefn\label{defyapcomplst}
A \emph{compensation list} is a list with $ n $ elements denoted as $ C =
[c_0, c_1, \ldots , c_{n-1} ] $, where each element $ c_i (x_0 , \ldots , x_{n-1}),\,\, i \in \{0, \ldots , n - 1\} $ is a combination of variants of shifted monomials.\\
The combination of two compensation lists $ C_{m 1} $ and
$ C_{m 2} $, for two monomials $ m 1 $ and $ m 2 $, is defined as the bitwise XOR $ C = C_{m 1} \oplus C_{m 2} $.
\edefn
Given an $ n $-bit NFSR, the compensation list for shifting a monomial $ m $ from $ f_a $ to $ f_b $ with $ a, b \in [0, n - 1] $ and $ a > b $,  is constructed as
\[C = [c_0, c_1, \ldots , c_{n-1} ] = [0, \dots, 0, m|_{-(a-b)} , \ldots, m|_{-1} , 0,\ldots , 0]\]
where $ c_i = 0 $ for $ i \in [0, b]\,\, \text{or}\,\, [a + 1, n - 1] $ and $ c_i = m|_{-(a-i+1)} $ for $  i \in [b + 1, a] $. 
\subsubsection{Compensation vs. Iterative Compensation}
In addition to the compensation operation, where the state variables are directly XOR-ed with the corresponding entries in the compensation list, Yao and Parampalli define the operation of \emph{iterative compensation in descending (or ascending) order}. We discuss this with an illustration, following \cite{yapa22}, and henceforth term the direct compensation operation as \emph{simple} compensation to emphasize the difference.\\
Let the indices $ x_i, x_j, \,  j=i-1, $ be compensated by the entries $ c_i, c_j $ of the compensation list, respectively.\\
Simple compensation results in $ x_i\rightarrow x_i\oplus c_i  $ and $ x_j \rightarrow x_j\oplus c_j $.\\
For iterative compensation \emph{in the descending order}, $ x_i $ is first transformed to $ x_i\oplus c_i $. If $ c_i = c_i (\ldots, x_j, \ldots) $, then, in the next step, $ x_i\rightarrow x_i \oplus c_i(\ldots, x_j\oplus c_j, \ldots) $, i.e. the effect of compensation at the next lower index $ j= i-1 $ is applied iteratively on the compensated $ x_i $.
\begin{multline*}
	f(x_0, x_1, x_2) \xrightarrow{x_2\rightarrow x_2\oplus c_2(x_1)} f(x_0,x_1,x_2\oplus c_2(x_1))\\
	\xrightarrow{x_1\rightarrow x_1\oplus c_1} f(x_0,x_1\oplus c_1,x_2\oplus c_2(x_1\oplus c_1))
\end{multline*}

%This immediately leads to the following observation:\\
%\emph{The effect of iterative compensation in the descending (ascending) order will be different from that of simple compensation only for those indices $ i $ whose compensation modifiers $ c_i $ depend on $x_j, \, j<i $ ($ j>i $).\\}
\subsubsection*{Yao-Parampalli Algorithms for Transformations between Fibonacci and Uniform Galois NFSRs \cite{yapa22}}
\begin{itemize}
	\item \emph{Fibonacci-to-Uniform Galois}\\
	\textbf{Input:} $ n $-bit Fibonacci NFSR; output function $ f_z $; initial state $ X^0 $.\\
	\textbf{Output:} $n$-bit uniform Galois NFSR; transformed output function $\overline{f_z}$;\\ matched initial state $\hat{X}^0$.
\begin{enumerate}
	\slen{1.65}
	\item For each monomial $ m_j , \, j\in [1,r] $, compute the lowest permissible shifting index $ p_j $ using Dubrova's algorithm in \cite{dubrova09} and choose a position $ \tau_j $, such that $ p_j \leq \tau_j \leq n-1 $.
	\item Shift the $ r $ monomials from feedback function $ f_{n-1} $ to $ f_{\tau_1}, f_{\tau_2} , \ldots, f_{\tau_r} $, where $ 0\leq \tau_1 \leq \tau_2\leq \ldots \leq \tau_r \leq n-1 $.
	\item For each monomial $ m_j, \, j \in [1,r] $, construct compensation list 
	\[C_{m_j} = [0, \ldots, 0, m_j|_{-(n-1-\tau_j)}, \ldots, m_j|_{-1}]\]
	XOR all such lists to compute a combined compensation list:\\
	 $ C = C_{m_1}\oplus  C_{m_2}\oplus \ldots \oplus C_{m_r}$.
	\item Compensate the output function $ f_z $ by $ C $ iteratively in the descending order to obtain the output function $ \overline{f_z} $ for the transformed NFSR.
	\item The initial state computed as: $ \hat{x}^0_i = x_i^0 \oplus c_i(X^0) , \,\, i \in [0,n-1]$.
\end{enumerate}
\item \emph{Uniform Galois-to-Fibonacci}\\
\textbf{Input:} $n$-bit uniform Galois NFSR; output function $ f_z $; initial state $ X^0 $.\\
\textbf{Output:} $ n $-bit Fibonacci NFSR; transformed output function $\overline{f_z}$; matched initial state $\hat{X}^0$.
\begin{enumerate}
	\slen{1.65}
	\item Shift every monomial $ g_i \neq 0 $, from $ f_i, \, i \in [\tau, n-2] $ to $ f_{n-1} $.
	\item For each $ g_i\neq 0 $ in the first step, compute the compensation list 
	\[C_{g_i} = [0, \ldots, 0, g_i, g_i|_{+1}, \ldots, g_i|_{+(n-2-i)}]\]
	XOR all such lists to compute a combined compensation list: 
	\[ C = C_{g_0}\oplus  C_{g_1}\oplus \ldots \oplus C_{g_{n-2}}.\]
	
	\item Compensate the output function $ f_z $ by $ C $ \emph{without iteration} to obtain the output function $ \overline{f_z} $ for the transformed Fibonacci NFSR as
	\[\overline{f_z}  = f_z(x_0\oplus c_0, \ldots, x_{n-1}\oplus c_{n-1})\]
	\item The initial state $ \hat{X}^0 = [\hat{x}^0_0, \hat{x}^0_1, \ldots, \hat{x}^0_{n-1}] $ is given by: 
	\[\hat{x}^0_i = x_i^0 \oplus \overline{c_i}(X^0) , \,\, i \in [0,n-1]\]
	where $ \overline{c_i} $ is obtained from $ c_i $ by iterative compensation in descending order using $ C $ .
\end{enumerate}
\end{itemize}
\brmk
The Yao-Parampalli algorithms directly use Dubrova's formulas in \cite{dubrova09} to define the shifted monomials as well as the terminal bits associated with each shift.
\ermk
We conclude this section with a summary comparison of the relevant features of the Dubrova and Y-P frameworks in the following Table \ref{tablepam}. This serves as an overture to the analyses and results in subsequent sections, in particular, Section \ref{secyapdubri}, where the connection between the frameworks is analysed further in the context of Fibonacci-to-uniform Galois transformations.
\begin{table}
	\centering
	\begin{tabular}{|c|p{3cm}|p{2.7cm}|p{1.8cm}|p{1.8cm}|}
		\hline
		\textbf{Framework}& \textbf{Transformation Profile}  & \textbf{Equivalence} &

		 \textbf{Matching Initial States}  & \textbf{Validation \newline Status}  \\
		\hline
	\multirow{3}{*}	{Dubrova} & Fibonacci to Uniform Galois \cite{dubrova09} & Output at final stage & No & Proved \\
			\cline{2-5}
	                        & Fibonacci to Uniform Galois \cite{dubrova10} & Output at final stage & Yes & Proved\\
	                        \cline{2-5}
	                        & Galois to Galois \cite{dubrova14} & Output function on restricted subset of stages & Yes & Proved\\
		
		\hline
			& & & &\\
	\multirow{2}{*}	{Yao-Parampalli} & Fibonacci to Uniform Galois \cite{yapa22} & Output function on any subset of stages & Yes & Claimed \\
	\cline{2-5}
	                               & Uniform Galois to Fibonacci \cite{yapa22} & Output function on any subset of stages & Yes & ------ \\
		\hline
	\end{tabular}
	\vspace{0.5 cm}
	\caption{Comparison of Dubrova and Yao-Parampalli Frameworks}
	\label{tablepam}
\end{table}
%\textbf{Observations:}
%\begin{enumerate}
%	\item The Yao-Parampalli algorithms directly use Dubrova's formulas in \cite{dubrova09} to define the shifted monomials as well as the terminal bits associated with each shift.
%	\item It is observed that the sequence of iterative compensation and simple compensation are interchanged in the algorithms proposed for Fibonacci-to-uniform Galois and the transformation in the reverse direction.
%	\item Theoretical validation is attempted \emph{only} in the case of Fibonacci-to-uniform Galois transformation (cf. Lemma 1 and Lemma 2 of Section 2, and Theorem 1 of Section 3 in \cite{yapa22}).
%\end{enumerate}
\section{``Fixing" Dubrova's Main Theorem}\label{secbhuldubri}
All the transformation algorithms analysed in the subsequent sections of this article build on Dubrova's Fibonacci-to-uniform Galois NFSR transformation in \cite{dubrova09}. However, we show that there is a non-trivial flaw in the proof of the key result - termed the ``main theorem'' in \cite{dubrova09}. In the remainder of this section we rectify this discrepancy and present a correct proof of this theorem, thereby validating a keystone result in Dubrova's transformation framework.\\
We begin by stating the result in question from \cite{dubrova09}:
\bthm\label{thmain}
Given a uniform NLFSR with terminal bit $ a $, a shifting $ g_a \xrightarrow{P} g_b $,  $ P\subseteq A_{g_a} $, $ b<a $, results in an equivalent NLFSR if the transformed NLFSR is uniform as well.
\ethm
The proof uses the \emph{feedback graph} description of an NFSR as defined below \cite{dubrova09}.
\bdefn
The \emph{feedback graph} of an $ n $-bit NFSR is a directed graph with $ n $ vertices $ v_0,\ldots, v_{n-1} $ which represents bits $ 0,\ldots, n-1 $ of the NFSR, respectively. There is an edge from $ v_i $ to $ v_j $ if $ i \in \text{dep}(f_j) $.
\edefn
It is noted that the discrepancies in the proof of this theorem stem from an incorrect application of the `substitution' operation (Definition 3 of Section III in \cite{dubrova09}), stated below.
\bdefn\label{defsubs}
The operation \emph{substitution}, denoted by $ \text{sub}(v_i, v_j) $, is defined for any vertex $ v_i $, which has a unique predecessor $ v_j $. The substitution $ \text{sub}(v_i, v_j) $ removes $ v_i $ from the feedback graph and, for each successor $ v_k $ of $ v_i $, replaces the edge $ (v_i, v_k) $ by an edge $ (v_j, v_k) $.
\edefn
\noindent Clearly the operation is defined as $ \textbf{sub}(v(\textbf{successor}), v(\textbf{predecessor}))$.
\subsection{The Discrepancy and Proposed Correction}
We first outline the relevant steps of the proof provided in \cite{dubrova09}. While we adhere to the original notation in \cite{dubrova09} as much as possible, we label the first uniform NFSR as $ N_a $ and the second as $ N_b $, according to the respective terminal bits. The NFSR transformation is effected by the shift of a single product-term $ x_k x_a, \, k<a $, from $ g_a $ to $ g_b $. Express $ g_a $ as $ g_a = g_a^{\ast} \oplus x_k x_a$.\\
The strategy of Dubrova's proof may be stated as follows:
\smallskip

\emph{If the NFSR $ N_b $, with terminal bit $ b $, is uniform, its feedback graph is reducible to a single vertex $ v_b $. Then there exists a non-linear recurrence describing the sequence at $ v_b $. The equivalence between $ N_a, N_b $ is established if this recurrence is shown to be identical with that describing the sequence at $ v_a $ for $ N_a $.}\\

\noindent \textbf{Relevant Steps in Dubrova's Proof:}\\
Let $ \tau = a $ and $ \tau = b $ represent the respective terminal bits for $ N_a $ and $ N_b $. As outlined in the steps of the proof to Theorem 1 of \cite{dubrova09}, the substitution operation is invoked twice in arriving at the non-linear recurrence for each terminal bit $ \tau, \tau \in \{a,b\} $. For instance, for $ \tau =a $, the vertices $ i<a$ are first removed via the substitution sequence:
	\beqn  \text{sub}(v_0, v_{1}), \text{sub}(v_1, v_{2}), \ldots, \text{sub}(v_{a-1}, v_{a})
	\eeqn
Then the vertices $ v_i, i>\tau =a$ are removed as follows in Section IV of \cite{dubrova09}, page 5268 (we quote in italics):
\smallskip

``\emph{After a	sequence of $ n-a-1 $ substitutions $ \text{sub}(v_{n-1}, v_{n-2}), \ldots, \text{sub}(v_{a+1}, v_{a}) $, we get a non-linear recurrence describing the sequence of values of bit $ a $}
\begin{multline*}
s_a(t) = s_a (t-n)\oplus g_{n-1}(\tilde{s}_a(-n+a+1))\oplus g_{n-2}(\tilde{s}_a(-n +a))\\
\oplus \cdots \oplus g_a^{\ast}(\tilde{s}_a)\oplus s_a(t-1-a+k)s_a(t-1). \text{''}
\end{multline*}
The notation $ \tilde{s}_a(i) $ has been defined (cf. pages 5267, 5268) to mean that each element $ s_a(x) $ of $ \tilde{s}_a $ has been replaced by $ s_a(x+i) $, where $ \tilde{s}_a $ stands for 
\beq\label{eq:satilda}
\tilde{s}_a:= (s_a(t-a-1), s_a(t-a), \ldots, s_a(t-1-(a-b)))
\eeq
A similar step is described in arriving at the non-linear recurrence for the values of the bit $ b $.\\
\noindent \textbf{The Discrepancy}\\
We observe that:
	\begin{itemize}
		\item The vertices $ i>\tau $, for $\tau =a $ or $b$, are removed via the sequence:
			\beq\label{eqfaltu}
			\text{sub}(v_{n-1}, v_{n-2}), \ldots, \text{sub}(v_{\tau +1}, v_{\tau})
			\eeq 
			It is noted that the arguments of the terms in Equation (\ref{eqfaltu}) are not consistent with the definition of the substitution operation (cf. Definition \ref{defsubs}) as they have the form \[(v(\text{predecessor}),v(\text{successor}) ). \] 
		\item Even if the order of the arguments is reversed in each operation, the second substitution sequence fails to achieve the presumed objective of reducing all the remaining vertices to $ v_{\tau} $.
\end{itemize} 
However, a simple modification to the strategy of applying the substitution operation suffices to remove this discrepancy in the proof, as discussed below.
\smallskip

\noindent\textbf{The Correction}\\
We aver that the correct strategy to evaluate the non-linear recurrence for the sequence of values of the terminal bit is to apply the following sequence of substitutions, \emph{after} all the vertices labelled $ i< \tau $ have been removed:
\beqn
\text{sub}(v_{n-1}, v_{\tau}), \text{sub}(v_{n-2}, v_{\tau}),\ldots, \text{sub}(v_{\tau +1}, v_{\tau})
\eeqn
Changing the time-shifts of $ s_{\tau}(t) $ appropriately at each stage will result in identical non-linear recurrences for both $ \tau = a $ and $ \tau =b $, thereby proving the theorem.
\subsection{Proposed Correct Proof of Dubrova's Theorem}
We now present our proof of Theorem \ref{thmain} in accordance with the correct formulation proposed in the previous subsection. We have retained Dubrova's notation from \cite{dubrova09} as much as possible and elaborate only those steps where there are significant modifications.
\smallskip

\noindent{\textbf{Proof of Theorem \ref{thmain}:}}
\smallskip

We adhere to Dubrova's proof in the setting of the theorem as well as the relevant steps in \cite{dubrova09} prior to the discrepancy. The transformation of the uniform NFSR $ N_a $, with terminal bit $ a $, into $ N_b $, with terminal bit $ b <a$, is effected by the shift of $ x_k x_a, \, k<a $, from $ g_a $ to $ g_b $. As the transformed $ N_b $ is uniform by hypothesis, and $ b<a $, the feedback components $ g_{n-1}, g_{n-2}, \ldots, g_a^{\ast} $ of $ N_a $ depend on variables $ s_j $ with $ j\leq b $. Accordingly, the state-variable values $ s_{n-1}(t), \ldots, s_0 (t) $ of $ N_a $ are described in terms of the state-variable values $ s_i(t-1) $ and the functions $ g_j(s_0(t-1), s_1(t-1), \ldots, s_b (t-1)) $.\\
For instance, \beq\label{eqNastat} 
s_{n-1}(t) = s_0(t-1)\oplus g_{n-1} (s_0(t-1), s_1 (t-1),\ldots, s_b (t-1)) 
\eeq
and the equation for the terminal bit of $ N_a $ is: 
\beq\label{eqNalast}
s_a(t) = s_{a+1}(t)\oplus g_{a}^{\ast} (s_0(t-1), s_1 (t-1),\ldots, s_b (t-1))\oplus s_k(t-1) s_{a}(t-1)
\eeq
Successive application of the substitutions $ \text{sub}(v_i, v_{i+1}) $ from the stage $ i = 0$ to the stage $i= a-1 $, results in the expression of all the remaining state-variable values as follows (\cite{dubrova09}).
\begin{align*}
	s_{n-1}(t) &= s_{a}(t-a-1)\oplus g_{n-1}(s_a(t-a-1), s_a(t-a), \ldots, s_a(t-1- (a-b))) \\
	s_{n-2}(t) &=  s_{n-1}(t-1)\oplus g_{n-2}(s_a(t-a-1), s_a(t-a), \ldots, s_a(t-1- (a-b)))  \\
	\cdots & =\cdots \\
	s_a(t) & = s_{a+1}(t-1)\oplus s_a (t-1-a+k)s_a(t-1)\\
	& \oplus g_{a}^{\ast}(s_a(t-a-1), s_a(t-a), \ldots, s_a(t-1- (a-b))) 
\end{align*}
The above sequence of substitutions results in replacing all $ s_i (t), \, i< a,$ with time-shifted $ s_a(t) $'s. Our next step involves the \emph{proposed second substitution sequence} $ \text{sub}(v_i, v_{a}) $ from the stage $ i = n-1$ to the stage $i= a+1 $.\\
The initial substitution in the second sequence, $ \text{sub}(v_{n-1}, v_{a}) $, has already been computed at the end of the first substitution sequence. This is because, at the end of the first sequence, $ v_a $ is the unique predecessor of $ v_{n-1} $ and all the $ s_i(t-1) $ have been replaced by $ s_a(t-a-1+i) $ for $ i\leq b< a $. Using the notation $ \tilde{s}_a $ from Equation \ref{eq:satilda}, we have:
\beqn
s_{n-1}(t) = s_a(t-a-1)\oplus g_{n-1}(\tilde{s}_a)
\eeqn
Next we arrive at the expression for $ s_{n-2}(t) $ as follows.
\begin{align*}
	s_{n-2}(t)& = s_{n-1}(t-1)\oplus g_{n-2}(\tilde{s}_a)\\
	& = s_a(t-a-2)\oplus g_{n-1}(\tilde{s}_a(-1))\oplus g_{n-2}(\tilde{s}_a)
\end{align*}
\emph{Claim:} At the $ i $-th stage of the second substitution sequence, $ i< n-a $, the state-variable value expression is given by:
\beq\label{eqinducsub}
s_{n-i}(t) = s_a(t-i-a)\oplus g_{n-1}(\tilde{s}_a(-i+1))\oplus g_{n-2}(\tilde{s}_a(-i+2)) \oplus \ldots \oplus g_{n-i}(\tilde{s}_a)
\eeq
\emph{Proof of Claim:} Proved by induction on $ i $, where $1\leq i\leq n-a-1 $.\\
Assuming Equation (\ref{eqinducsub}) holds for some $ i =m-1 $, we have, for $ i =m $,
\[s_{n-m}(t) = s_{n-m+1}(t-1)\oplus g_{n-m}(\tilde{s}_a)\]
Substituting the formula for $ s_{n-m+1}(t-1) $ proves the claim.\\
\smallskip

\noindent Applying the above formula for $ i = n-a-1 $, we have
\begin{align*}
	s_{a+1}(t) & = s_a(t-(n-a-1)-a)\oplus g_{n-1}(\tilde{s}_a(-n+a +2)) \oplus \ldots \oplus g_{a+1}(\tilde{s}_a)\\
	& = s_a(t-n +1)\oplus g_{n-1}(\tilde{s}_a(-n+a +2)) \oplus g_{n-2}(\tilde{s}_a(-n+a +3)) \oplus \\
	& \hspace{0.2 in}\ldots \oplus g_{a+1}(\tilde{s}_a)
\end{align*}

\noindent Therefore
\begin{align*}
	s_a(t) & = s_{a+1}(t-1)\oplus s_a (t-1-a+k)s_a(t-1)\\
	& \oplus g_{a}^{\ast}(s_a(t-a-1), s_a(t-a), \ldots, s_a(t-1- (a-b)))\\
	& = s_a(t-n)\oplus g_{n-1}(\tilde{s}_a(-n+a+1))\oplus g_{n-2}(\tilde{s}_a(-n+a +2))\oplus \ldots\\
	& \oplus g_{a+1}(\tilde{s}_a(-1))\oplus g_{a}^{\ast}(\tilde{s}_a)\oplus s_a (t-1-a+k)s_a(t-1) 
\end{align*}
\noindent Expanding the abbreviated notation for $ \tilde{s}_a(-i) $ in each term we get:
\beq\label{eqsaful}
\begin{aligned}
	s_a(t) & = s_a(t-n)\\
	&\oplus g_{n-1}(s_a(t-n), s_a(t-n+1), \ldots, s_a(t-n+b))\\
	&\oplus g_{n-2}(s_a(t-n+1), s_a(t-n+2), \ldots, s_a(t-n+b+1))\oplus \ldots\\
	&\oplus g_{a}^{\ast}(s_a(t-a-1), s_a(t-a), \ldots, s_a(t-1-(a-b)))\\
	&\oplus s_a (t-1-a+k)s_a(t-1) 
\end{aligned}
\eeq
\smallskip

\noindent For the transformed NFSR $ N_b $, the state-variable values are given by the following set of equations (cf. \cite{dubrova09}):
\begin{align*}
	s_{n-1}(t) &= s_{0}(t-1)\oplus g_{n-1}(s_0(t-1), s_1(t-1), \ldots, s_b(t-1)) \\
	\ldots &= \ldots\\
	s_{a}(t) &=  s_{a+}(t-1) \oplus g_{a}^{\ast}(s_0(t-1), s_1(t-1), \ldots, s_b(t-1)) \\
	\ldots & =\ldots \\
	s_{a-1}(t) & = s_{a}(t-1)\\
	\ldots &= \ldots\\
	s_b(t) & = s_{b+1}(t-1)\oplus s_{k-a+b} (t-1)s_b(t-1)\\
	\ldots & = \ldots\\
	s_0(t)& = s_1(t-1)  
\end{align*}
It is noted that the product-term $ x_k x_a $ has been shifted from $ s_a(t) $ to $ s_b(t) $ with the appropriate change in the indices: $ k\rightarrow k-a+b, a\rightarrow b $.\\
\noindent Using the notation (\cite{dubrova09})
\[\tilde{s}_b:= (s_b(t-b-1), s_b(t-b), \ldots, s_b(t-1))\]
we obtain, as follows (cf. \cite{dubrova09}), a similar set of equations after the first sequence of substitutions $ \text{sub}(v_i, v_{i+1}) $, for $ i= 0, \ldots, b-1 $.
\begin{align*}
	s_{n-1}(t) & = s_b(t-b-1)\oplus g_{n-1}(\tilde{s}_b)\\
	\ldots & = \ldots\\
	s_{a}(t) & = s_{a+1}(t-1)\oplus g_{a}^{\ast}(\tilde{s}_b)\\
	s_{a-1}(t)& = s_a(t-1)\\
	\ldots & = \ldots\\
	s_b(t) & = s_{b+1}(t-1)\oplus s_b(t-1-a+k)s_b(t-1)
\end{align*}
where, in the last line, the argument in the first term of the product-term on the r.h.s. is obtained thus: $ t-1 -(b -(k-a+b)) = t-1-a+k $; i.e. it is \emph{implicitly assumed that $ k-a+b \leq b $}.\\

\noindent However, as demonstrated earlier for the NFSR $ N_a $, the second sequence of substitutions to compute the recurrence for the bit $ b $ should be
\beqn
\text{sub}(v_{n-1}, v_{b}), \text{sub}(v_{n-2}, v_{b}),\ldots, \text{sub}(v_{b +1}, v_{b})
\eeqn
With a development similar to Equation (\ref{eqinducsub}), we present the corrected form for $ s_b(t) $ as follows.
\begin{align*}
	s_b(t) & = s_b(t-n)\oplus g_{n-1}(\tilde{s}_b(-n+b+1)) \oplus g_{n-2}(\tilde{s}_b(-n+b +2))\oplus \ldots\\
	& \oplus g_{a}^{\ast}(\tilde{s}_b(-a+b))\oplus s_b(t-1-a+k)s_b(t-1) 
\end{align*}
In the recurrence for $ s_b(t) $, the term $ g_a^{\ast} $ is contributed from the stage $ a $, and so, the argument of $ g_a^{\ast} $ becomes $ \tilde{s}_b(-(a-b)) $. This is in contrast with the term $ g_a^{\ast}(\tilde{s}_a) $ which appears in the final stage of the recurrence for $ s_a(t) $. However, substituting the full form of $ \tilde{s}_b $ in the $ g_a^{\ast} $ term in $ s_b(t) $ yields:
\[g_{a}^{\ast}(s_b(t-a-1), s_b(t-a), \ldots, s_b(t-1-(a-b)))\]
Hence, expanding the notation for all the $ \tilde{s}_b(-i) $'s in the expression of $ s_b(t) $ we obtain:
\beq\label{eqsbful}
\begin{aligned}
	s_b(t) & = s_b(t-n)\\
	&\oplus g_{n-1}(s_b(t-n), s_b(t-n+1), \ldots, s_b(t-n+b))\\
	&\oplus g_{n-2}(s_b(t-n+1), s_b(t-n+2), \ldots, s_b(t-n+b+1))\oplus \ldots\\
	&\oplus g_{a}^{\ast}(s_b(t-a-1), s_b(t-a), \ldots, s_b(t-1-(a-b)))\\
	&\oplus s_b (t-1-a+k)s_b(t-1) 
\end{aligned}
\eeq
Comparing Equations (\ref{eqsaful}) and (\ref{eqsbful}), it is readily seen that the recurrence for the bit $ a $ in NFSR $ N_a $ is identical with that for $ b $ in NFSR $ N_b $.\\

\noindent This completes the proof.
\qed
\brmk
It is noted that the final expression of the recurrence for the bit $ a $ in the original NFSR, given as Equation (7) in \cite{dubrova09} (p. 5268), and its abbreviated form using the notation $ \tilde{s}_a $, are both \textbf{erroneous}. The argument of the term $ g_{n-2} $ is shown in Equation (7) to be:
\[s_a(t-n-1), s_a(t-n), \ldots, s_a(t-n+b-1)\] 
when the correct argument is (cf. Equation \ref{eqsaful}):
\[s_a(t-n+1), s_a(t-n+2), \ldots, s_a(t-n+b+1)\]
The same error in the argument is present in the corresponding expressions for the bit $ b $ in the transformed NFSR (Equation (10), p. 5268 in \cite{dubrova09}) as well, and the correct expression is presented here in Equation (\ref{eqsbful}).
\ermk
To conclude, we have established that there are serious errors in the proof of E. Dubrova's theorem in \cite{dubrova09} and presented a correct proof. It is important to note that the statement of the theorem, which is a key result in numerous subsequent works by Dubrova and others, has been verified to be true - only the validation was flawed. In fact, the idea of transforming a uniform Galois NFSR via shifts of feedback functions and relating the notion of equivalence of NFSRs with equality of recurrence relations at designated stages have proved to be a powerful toolkit for analysis of shift register-based ciphers. 
\section{Unifying Dubrova and Yao-Parampalli Transformation Frameworks}\label{secyapdubri}
As indicated in the introductory section, there is a strong connection between the transformation algorithms defined in the works of Dubrova (\cite{dubrova09}, \cite{dubrova10}, \cite{dubrova14}) and those of Yao and Parampalli (\cite{yapa21},\cite{yapa22}). In this section, we attempt to formalize these connections in the context of transformations between Fibonacci and uniform Galois NFSRs. The proofs of most of the results used in this section have been presented earlier in \cite{space25}.
\subsection{Unifying Dubrova's Algorithms}\label{seceqdubri}
In this subsection we state the evolutionary relation among the transformation algorithms discussed in \cite{dubrova09} and \cite{dubrova14}. We restrict our discussion to the type of NFSRs featured in the Espresso proposal as captured in the following definition.
\bdefn\label{defstrict}
A uniform NFSR is termed a \emph{strict-terminal uniform NFSR}, denoted st-uniform NFSR, if it satisfies the following condition for all stages indexed by $ i>\tau $:
		\beqn
		\text{max\_index}(g_i)< \tau
		\eeqn
where $\tau$ is the terminal bit.
\edefn
For instance, the Espresso NFSR $G$ has terminal bit $\tau = 193$ with the largest $\text{max\_index}(g_i) = 189 $, while the Espresso alternate NFSR $F$ has terminal bit $\tau =217$, with $\text{max\_index}(g_i)\leq 183$. Hence they are both st-uniform NFSRs.\\
%In what follows we show that:
%\smallskip
%
%\emph{Dubrova's Fibonacci-to-uniform Galois transformation algorithm \cite{dubrova09} D1, restricted to st-uniform Galois NFSRs, is subsumed into the framework of the Galois-to-Galois algorithm D2 \cite{dubrova14}.}
%\smallskip

The first algorithm D1 in \cite{dubrova09} effects a transformation between a uniform Fibonacci to a uniform, fully shifted Galois NFSR, which shifts all the monomials from the stage $ n-1 $ to a stage with index greater than or equal to the terminal bit $ \tau $. The rules of this shifting are given in \cite{dubrova09} as follows.\\
Let $ \text{max\_index}(m) $ denote the largest index of variables in $ \text{dep}(m) $; likewise $ \text{min\_index}(m) $ denotes the smallest index. For all $ m \in S_{g_{n-1}} $, where $ S_{g_{n-1}} $ is the set of monomials of $ g_{n-1} $,
\begin{enumerate}
	\item If $ \text{min\_index}(m) \leq (n-1) - \tau $, then the monomial $ m $ is shifted as
	\[g_{n-1} \xrightarrow{m} g_{n-1-\text{min\_index}(m)}\]
	\item If $ \text{min\_index}(m) > (n-1) - \tau $, then the monomial is shifted to $ g_{\tau} $.
\end{enumerate}
The following result from \cite{space25} shows that the above method in the D1 algorithm for generating a fully shifted st-uniform Galois NFSR from a (uniform) Fibonacci NFSR, which preserves the output sequence at the $ 0 $-th stage, is consistent with a valid shifting in the D2 algorithm \cite{dubrova14}.
\begin{theorem}\label{thmdubridubri}
	Let $ F $ be a Fibonacci NFSR with feedback functions of the form given in Equation \ref{eqfnuni}, with the set of feedback monomials at the $ (n-1) $-th stage denoted by $ \{x_0\}\cup S_{g_{n-1}} $. Then the shift of any monomial $ m \in S_{g_{n-1}} $ under the algorithm D1 for st-uniform NFSRs is a valid shifting in the sense of the algorithm D2.
\end{theorem} 
\brmk Theorem \ref{thmdubridubri} shows that for Fibonacci to st-uniform Galois transformation, D1 in \cite{dubrova09} is subsumed in the framework of D2 in \cite{dubrova14}. Hence the post-transformation equivalence conditions (cf. Theorem \ref{thdubri2}) that are satisfied for D2 will hold for D1 as well.
\ermk
\subsubsection{Dubrova's Algorithms and the Espresso Equivalence Problem}\label{subseep} 
As discussed in the introductory section, an alternate Galois NFSR $F$ is defined in \cite{dubrova17} to discuss the security of Espresso, which is claimed to be ``equivalent'' to the cipher design. However, we point out that there are two aspects to any meaningful equivalence, in the context of security analysis, of the two NFSR configurations so defined.
\begin{enumerate}
	\item Establishing the equivalence of the generated keystreams at designated stages of the $ 2 $-tap Galois NFSR $ F $ and the Espresso NFSR $ G $.
	\item Defining a corresponding non-linear output function for $ F $, so that both systems generate \emph{identical keystreams}. No such function is provided in \cite{dubrova17}.
\end{enumerate}
These unresolved issues constitute the \emph{Espresso Equivalence Problem}, graphically depicted in Figure \ref{fig:dubrigupi}.
\begin{figure}[H] 
	\centering
		
	\tikzstyle{block} = [draw, rectangle, rounded corners, inner sep=6pt]
	\tikzstyle{joint} = [draw, circle,minimum size=1em]
	
	\begin{tikzpicture}[>=stealth, auto, node distance= 2.8 cm]
	% Place nodes
	\node [block] (G)[draw, align=center]{Espresso (Galois) NFSR $ G $\\ $14$ non-linear taps, 256 stages\\
		$20$-variable Output Function};
	\node [block, right=of G] (F)[draw, align=center] {Alternate (Galois) NFSR $ F $\\$2$ non-linear taps, 256 stages\\
	No Output Function defined};
%	\node [coordinate, right=of outfork] (output) {};
%%	\node [coordinate, above=of system] (disturbances) {};
%	\node [block, below=of system] (model) {Model};
%	\node [joint, right=of model, anchor=center,label={[shift={(2mm,-1mm)}]-},label={[shift={(-3mm,-5.5mm)}]\tiny +}] (sum) {};
%	
%	\node [coordinate, right=of sum] (error) {};

	% Connect nodes

	\draw [dotted, <->, line width= 0.54 mm] (G) -- node{Equivalent?} (F);
%	\draw [->] (system) -- node {\(t'\)} (output);
%	\draw [->] (model) -- node {\(y\)} (sum);
%	\draw [->] (sum) -- node {\(\epsilon\)} (error);
%	\draw [->] (infork) |- node [anchor=south west] {\(u\)} (model);
%	\draw [->] (outfork) -| (sum.north) node [very near end] {\(t\)};
	\end{tikzpicture}
			
	\caption{The Espresso Equivalence Problem}
	\label{fig:dubrigupi}
\end{figure}
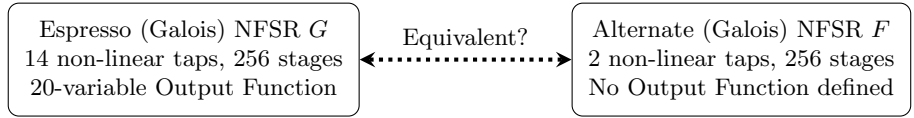

\noindent In \cite{dubrova17}, the only pointer to establishing the claimed equivalence is a reference, on p. 278, to Dubrova's algorithm D1 in \cite{dubrova09} which we quote below:\\

\noindent \emph{``The equivalence of G and F can be shown by applying the Fibonacci-to-Galois transformation [6].''}\\

\noindent In the light of Theorem \ref{thmdubridubri}, we now comment on the applicability of Dubrova's algorithms to the Espresso equivalence problem using the following Table 
\ref{tabdubridubri} (cf. \cite{space25}), which displays the relevant information culled from Table \ref{tablepam} in place for convenience. It is noted that Dubrova's first algorithm in \cite{dubrova09}, while introducing several key notions of the shift formula, uniform NFSR configuration and associated terminal bit, does not provide explicit formulation for matching initial state. Further, it guarantees equivalence across the transformation only for the output at a single stage (e.g. stage $0$). Hence, it has been argued in \cite{space25} that, contrary to the claim in \cite{dubrova17}, the algorithm D1 cannot be used to resolve the Espresso equivalence problem, while it may be possible to use the algorithm D2 to gain further insights into the problem. Further analysis in later sections confirms this initial observation.
\begin{table}
    \centering
    \begin{tabular}{|c|p{3cm}p{2.5cm}p{2.5cm}|}
    \hline
       \textbf{Algorithm}  & \textbf{NFSR Types} & \textbf{Multi-stage \newline Equivalence}  & \textbf{Matching \newline Initial States}  \\
       \hline
        D1 \cite{dubrova09} & Fibonacci to \newline Uniform Galois & No & No \\
        & &  &  \\
         D2 \cite{dubrova14}& Galois to Galois  & Yes & Yes \\
         \hline
    \end{tabular}
    \vspace{0.5 cm}
    \caption{Comparison of Dubrova's NFSR Transformation Algorithms}
    \label{tabdubridubri}
\end{table}
    \vspace{-0.9 cm}
\subsection{Relating Yao-Parampalli and Dubrova's Algorithms}
As discussed in Subsection \ref{subs:yap}, the Yao-Parampalli transformation framework, henceforth denoted as the Y-P framework, analyzes four configurations of Galois NFSRs and back-and-forth transformations with ``equivalent'' Fibonacci NFSRs, when feasible. However, the only proven technical results in \cite{yapa22} deal with the transformation of a Fibonacci NFSR to a uniform Galois NFSR, which is the precisely the case treated in \cite{dubrova09}. Accordingly, we compare the Y-P Fibonacci-to-uniform Galois algorithm with Dubrova's generalized transformation D2 restricted to this particular case.\\
The three proven results in \cite{yapa22} are as follows:
\begin{itemize}
	\item In Section 2 of \cite{yapa22}, Lemma $ 1 $ and Lemma $ 2 $ relate the internal states of an ``original Fibonacci NLFSR'' with that of the ``transformed Type-II Galois NLFSR''.
	\item The third result is Theorem 1 in Section 3 of \cite{yapa22}, which states that for ``an $ n $-bit Fibonacci NLFSR transformed to a Type-II Galois NLFSR according to the Fibonacci-to-Type-II algorithm'', the internal states are related by the relations in Lemma 1 and Lemma 2, and, further, the NFSRs are ``equivalent''.
\end{itemize}  
It was proved in \cite{space25} that matching the initial states before and after the transformation in the Y-P framework yields the same result as outlined in Dubrova's algorithm D2. In the following, we briefly state the relevant results in \cite{space25} and further extend the analysis.
\subsubsection{Matching Initial States}
We begin by stating the two lemmas, proved in \cite{yapa22}, and compare them with Dubrova's Theorem \ref{thdubri2}. For facilitating this comparison, we use the following\\
\emph{Notation}: Denote the $ k $-th stage at time $ t $ of the transformed uniform Galois NFSR as $ r_k^t $, and let $ s_k^t $ denote the corresponding entity for the Fibonacci NFSR. The complete states of the said NFSRs at time $ t $ are denoted with the boldface $ \rb^t $ and $ \bfs^t $, respectively.\\
The lemmas establish that for the shift of a single monomial, using Dubrova's shift formula in \cite{dubrova09}, a relation between the initial states of the Fibonacci and the transformed uniform Galois NFSRs is preserved for all subsequent states.
\blem[Lemma 1, \cite{yapa22}]\label{lemlyapa1}
Given an $ n $-bit Fibonacci NFSR, a monomial $ m $ is shifted from $ f_{ n-1} $ to $ f_b $ according to the shift operation in Definition \ref{defshft} and the resulted feedback functions satisfy the conditions for uniform Galois NFSRs. If the initial state of the uniform Galois NFSR is computed according to (\ref{eqlyapa1}), then the internal state of the uniform Galois NFSR $ r^t_k, k \in
\{0, \ldots, n-1\} $ and the internal state of the Fibonacci NFSR $ s^t_k, k \in
\{0, \ldots, n-1\} $ also satisfy this relation for $ t > 0 $.
\begin{align}\label{eqlyapa1}
\begin{split}
r_k & = s_k \oplus m |_{-(n-k)}(\bfs^t),\,\,\text{for}\,\, k \in [b+1, n-1]\\
r_k & = s_k,\,\,\text{for}\,\, k \in [0,b]
\end{split}
\end{align}
\elem
\smallskip
\noindent The second lemma of the Y-P framework, which features iterative compensation, is stated as follows.
\blem[Lemma 2, \cite{yapa22}]\label{lemlyapa2}
In the setting of Lemma \ref{lemlyapa1},
\begin{align}\label{eqlyapa2}
	\begin{split}
		s_k & = r_k \oplus \overline{m |_{-(n-k)}}(\rb^t),\,\,\text{for}\,\, k \in [b+1, n-1]\\
		s_k & = r_k,\,\,\text{for}\,\, k \in [0,b]
	\end{split}
\end{align}
where $ \overline{m |_{-(n-k)}} $ is $ m $ compensated by $ C =  [0, \dots, 0, m|_{-(a-b)} , \ldots, m|_{-1} , 0,\ldots , 0] $ iteratively in descending order.
\elem
The following result in \cite{space25} established the equivalence of matching the initial states in Y-P and Dubrova's formulations by proving the equivalence of Lemma 1 of \cite{yapa22} with Dubrova's theorem in \cite{dubrova14}.
\blem\label{lemyapchor1}
For a Fibonacci-to-uniform Galois transformation, the formulation in \cite{yapa22} for the Galois NFSR stages $ r_k^t $ in terms of the Fibonacci stages $ s_k^t $, $ k \in
\{0, \ldots, n-1\} $, is simply obtained by a restriction of the formulation in Dubrova's Theorem in \cite{dubrova14} to this particular case.
\elem
It is noted that even though the results of matching the initial states are identical, when the transformed initial state is expressed in terms of the original state, there is no obvious counterpart of Lemma \ref{lemlyapa2} in Dubrova's general framework. Further, the Y-P framework introduces the compensation operation to arrive at the results, which has no apparent equivalent in Dubrova's works. In the next subsection, we take a closer look at the actual analyses behind these results and reveal a stronger connection between the Dubrova and Y-P frameworks.
\subsection{De-mystifying Compensation}\label{subscompenis}
In this subsection we first examine the two compensation operations defined in the Y-P framework in the context of Fibonacci-uniform Galois transformations, and show that they lead to the evaluation of the same matching initial states, albeit from different directions. Then we prove that the simple compensation operation is a re-formulation of a result from Dubrova's framework \cite{dubrova10}, which allows us to relate these two frameworks with greater clarity.
\subsubsection{Simple vs. Iterative Compensation} It is evident from the statements of Lemma \ref{lemlyapa1} and Lemma \ref{lemlyapa2} (cf. Lemma 1 and Lemma 2 of \cite{yapa22}) that the second lemma includes the first in its premise. We further make the following formal observation, which is obtained from the steps of the proof of the second lemma presented in \cite{yapa22}.
\bprop
For the Yao-Parampalli Fibonacci-to-uniform Galois transformation, the initial states before and after transformation can be matched in either of two ways, leading to identical results, as follows:
\begin{enumerate}
	\item Start with the initial state of the original Fibonacci NFSR and perform simple compensation to obtain the matching initial state of the transformed Galois NFSR;
	\item Start with the initial state of the transformed Galois NFSR and perform iterative compensation to obtain the matching initial state of the original Fibonacci NFSR.
\end{enumerate}
\eprop
\begin{proof}
It follows from Equation \ref{eqlyapa1} and Equation \ref{eqlyapa2} that for indices $k \in [0,b]$, where $k = b+1$ is the first possibly non-zero index in the compensation list, both the procedures return identical values by default.\\
For any $k\in [b+1, n-1]$, consider the expression for the $k$-th initial bit of the Fibonacci NFSR in terms of the corresponding bit of the transformed Galois NFSR obtained from Equation \ref{eqlyapa2}:
\[	s^0_k = r^0_k \oplus \overline{m |_{-(n-k)}}(\rb^0)\] 
Here, the overline denotes iterative compensation using the initial state $\rb^0$ of the Galois NFSR. But, as Lemma 2 pre-supposes Lemma 1, we substitute the value of $r^0_k$ from Equation \ref{eqlyapa1} in the above to obtain:
\beqn
s^0_k = (s^0_k \oplus m |_{-(n-k)}(\bfs^0))\oplus \overline{m |_{-(n-k)}}(\rb^0)
\eeqn
which readily implies 
\beq\label{eqcompenis}
	m |_{-(n-k)}(\bfs^0) = \overline{m |_{-(n-k)}}(\rb^0)
\eeq
 This proves the assertion, as simple compensation using the initial state of the Fibonacci NFSR yields exactly the same values as iterative compensation using the state of the transformed Galois NFSR.
\end{proof}
\qed
\brmk
It is noted that the proof of Lemma 2 in \cite{yapa22} essentially consists of proving the relation in Equation \ref{eqcompenis} (cf. pp. 235-236 \cite{yapa22}).
\ermk
 
\subsubsection{Simple compensation and Dubrova's matching formula in \cite{dubrova10}}
As discussed in Section \ref{shegoldubri}, formulas for matching initial states in transformations between Fibonacci and uniform Galois NFSRs were presented in \cite{dubrova10}, \cite{dubrova14}. Specifically, Theorem \ref{thm:compendubri} in \cite{dubrova10}, gives the explicit formula for the initial state of a uniform Galois NFSR in terms of the initial state of an ``equivalent'' Fibonacci NFSR. In the following we show that, in this case, simple compensation in the Y-P framework is a re-formulation of the above formula in \cite{dubrova10}.
\bthm\label{thm:chorcompenis}
Let $N_F$ and $N_G$ be an equivalent pair, in the sense of matched outputs at stage $0$, of an $n$-bit Fibonacci and a uniform Galois NFSR, respectively. Then matching the initial state $\bfs= (s_0, \ldots, s_{n-1})$ of $N_F$ to that of $N_G$, denoted $\rb = (s_0, s_1, \ldots, s_{\tau}, r_{\tau+1}, \ldots,r_{n-1} )$, through the Yao-Parampalli simple compensation operation yields the formula:
\[r_k = s_k \oplus \bigoplus_{i=\tau}^{k-1}g_{i}|_{+(k-(i+1))}(\bfs)\]
which matches exactly with the corresponding formula for $r_k, \, \, \tau +1\leq k \leq n-1,$ in (Dubrova's) Theorem \ref{thm:compendubri}.
\ethm
\begin{proof}
	Consider the case where the uniform Galois NFSR $N_G$ has terminal bit $\tau$, where $0\leq\tau< n-1$, and feedback functions of the form:
	\begin{align*}
		f_k(x) & = x_{(k+1)\!\!\!\!{\mod n}}\oplus g_k (x_0, \ldots, x_{\tau}); \,\, k \in [\tau, n-1]\\
		f_k(x) & = x_{(k+1)}; \,\, k\in [0,\tau -1]
	\end{align*}
	Hence, the transformation entails the \emph{shift} of all feedback functions, which themselves consist of XORed monomial terms, from the single Fibonacci feedback function at $k = n-1$ of $N_F$. The terminal bit being $\tau$, the first possibly non-zero entry in the compensation list corresponds to the shift from $n-1$ to $\tau$, with the \emph{shifted form} being $g_{\tau}$. We denote the original component of the feedback at the $(n-1)$-th stage of $N_F$ as $g^{\ast}_{\tau}$, which is related to the shifted feedback function $g_{\tau}$ of $N_G$ as follows:
	\[g_{\tau} = g^{\ast}_{\tau}|_{-(n-1-\tau)} \implies g^{\ast}_{\tau} = g_{\tau}|_{+(n-1-\tau)}\]
	The compensation list corresponding to the shift $f_{n-1}\xrightarrow{g^{\ast}_{\tau}}f_{\tau}$ is given by:	
	\[C_{\tau} =[0,\ldots, 0, \underbrace{g^{\ast}_{\tau}|_{-(n-1-\tau)}}_{\tau +1}, g^{\ast}_{\tau}|_{-(n-2-\tau)}, \ldots, \underbrace{g^{\ast}_{\tau}|_{-1}}_{n-1}]\]
	Substituting $g_{\tau}$ for $g^\ast_{\tau}$ in the above, we have:
	\[C_{\tau} =[0,\ldots, 0, \underbrace{g_{\tau}}_{\tau +1}, g_{\tau}|_{+1}, \ldots, \underbrace{g_{\tau}|_{+(n-1-(\tau +1))}}_{n-1}]\]
	which gives the simple compensation contribution involving $g_{\tau}$ in the combined compensation list: 
	\[C = C_{\tau} \oplus C_{\tau +1}\oplus \cdots \oplus C_{n-1}\]
	It can be similarly verified that the compensation contribution involving $g_{k}$, for all $\tau +1 \leq k \leq n-1$, is given by:
	\[C_{k} =[0,\ldots, 0, \underbrace{g_{k}}_{k +1}, g_{k}|_{+1}, \ldots, \underbrace{g_{k}|_{+(n-1-(k +1))}}_{n-1}]\]
	The first possibly non-zero entry of the combined compensation list is the entry in the position $\tau +1$, given by the single term $g_{\tau} $.\\
	The entry in the position $\tau +2$ is accordingly given by: $g_{\tau +1}\oplus g_{\tau}|_{+1}$.\\
	Collecting the terms in a similar manner, the $k$-th entry in the combined compensation list is given by:
	\[g_{k-1}\oplus g_{k-2}|_{+1}\oplus \cdots \oplus g_{\tau}|_{+(k-(\tau +1))}= \bigoplus_{i=\tau}^{k-1}g_{i}|_{+(k-(i+1))}\]
	Simple compensation using this combined compensation list yields the $k$-th bit of the initial state of $N_G$, say $r_k$, in terms of the initial state $\bfs$ of $N_F$ as:
	\[r_k = s_k \oplus \bigoplus_{i=\tau}^{k-1}g_{i}|_{+(k-(i+1))}(\bfs).\]
	Hence the theorem.	
\end{proof}
\qed
We summarize our analysis of the transformation frameworks of Dubrova and Yao-Parampalli with the following\\

\noindent \textbf{Observations:}
\begin{enumerate}
	\item For Fibonacci-to-uniform Galois NFSR transformation, the simple compensation operation of Yao-Parampalli yields the same result as Dubrova's formulation for matching initial states to achieve NFSR equivalence.
	\item In the Y-P framework for Fibonacci-uniform Galois transformation, simple compensation and iterative compensation yield identical pairs of matched initial states. However, simple compensation computes the transformed state in terms of the original state while iterative compensation proceeds in the reverse direction. 
	\item Therefore, with regard to equivalent transformations between Fibonacci and uniform Galois NFSRs, Dubrova's general framework (\cite{dubrova10}, \cite{dubrova14}) and the Y-P framework are essentially identical in computing the monomial shifts and matching the initial states. 
	\item However, as will be elaborated on in Section \ref{secbom}, the real efficacy of iterative/simple compensation is demonstrated in matching the outputs of non-linear functions, involving several state variables of the NFSRs, before and after the transformations. This contribution in the Y-P framework appears to represent a significant step beyond Dubrova's framework.
\end{enumerate}

\section{Resolving the Espresso Equivalence Problem}\label{secbom}
In this section we bring to bear our unifying analysis of Dubrova and Y-P transformation techniques on the Espresso Equivalence Problem (EEP) (cf. Subsection \ref{subseep}). As mentioned earlier, Yao and Parampalli have claimed in \cite{yapa21} that the application of \emph{uniform Galois-to-Fibonacci} NFSR transformation on the Espresso NFSR $ G $ results in a Fibonacci LFSR with feedback function
\beq\label{eqyapfol}
f_L(\xb):= {f}_{255}(x) = x_0 \oplus x_{12} \oplus x_{48} \oplus x_{115} \oplus x_{133} \oplus x_{213}
\eeq
This also generates an output function, which we denote as $ z_{\text{YP}}(\xb) $, consisting of $ 2289 $ monomials and having algebraic degree $ 12 $, such that the ensemble comprising LFSR $ f_L(\xb) $ and $ z_{\text{YP}}(\xb) $, with initial state matched with $ G $, produces \emph{the same keystream as Espresso}. This is indeed a fascinating result, and it showcases the power of Dubrova-Yao-Parampalli transformation framework by reducing a Galois NFSR with $14$ non-linear taps to a Fibonacci LFSR whose single feedback consists of $6$ linear terms. The Y-P LFSR-based Espresso-equivalent design is the springboard for our attempt to resolve EEP.\\
Thus far, we have shown that the NFSR transformation effected by the Y-P Fibonacci-to-uniform Galois method is essentially the same as Dubrova's algorithm D2. In \cite{space25}, the Fibonacci-to-uniform Galois algorithm was applied on the Y-P LFSR $ (f_L(\xb)) $ to attempt a transformation to the Espresso alternate NFSR $ F $. The rationale for this was given in \cite{space25} as follows:\\

\noindent \emph{If the NFSR $ F $ is recovered, then, given the equivalence claimed in \cite{yapa21}, the output function obtained using the Y-P algorithm on the LFSR ensemble would complete the gap in the Espresso NFSR $ F $ design.}\\

\noindent A concept schematic of the above strategy is given in Figure \ref{fig:akbebhul}.
\begin{figure}[t] 
	\centering
	\includegraphics[width= \textwidth]{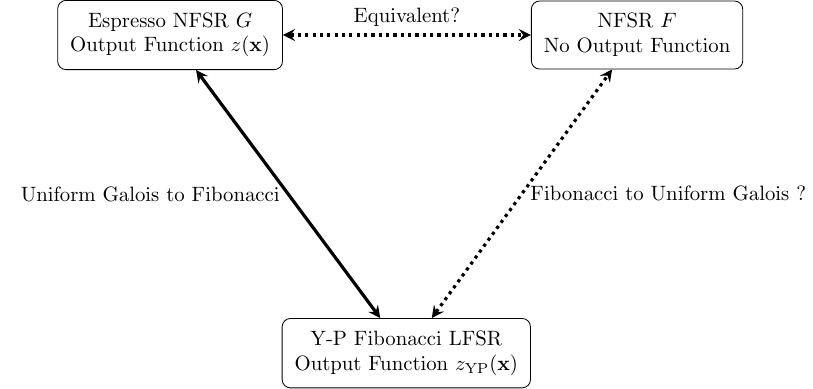}		
	\caption{Using Y-P/Dubrova Algorithms to Resolve Espresso Equivalence}
	\label{fig:akbebhul}
\end{figure}

\subsection{Recovery of Espresso NFSR $ F $ from Y-P LFSR}\label{subs:lofornofor}
We now present, for the sake of completeness, the steps of the recovery of the Galois NFSR $F$ by the application of the Fibonacci-to-uniform Galois algorithm on the Fibonacci LFSR $ f_L(\xb) $ as given in \cite{space25}.\\
We first note that the NFSR $ F $ has two feedback taps given by:
	\begin{align}\label{eqfespes}
\begin{split}
f_{255}(\xb) = & x_{0}\oplus x_{12}\oplus x_{48}\oplus x_{115}\oplus x_{133}\oplus x_{213}\oplus x_{41}x_{70}\oplus x_{46}x_{87}\\
&\oplus x_{52}x_{110}\oplus x_{55}x_{130}\oplus x_{62}x_{157}\oplus x_{74}x_{183}\oplus x_{87}x_{110}x_{130}x_{157}\\
f_{217}(\xb) = & x_{218}\oplus x_{3}x_{32}\oplus x_{8}x_{49}\oplus x_{14}x_{72}\oplus x_{17}x_{92}\oplus x_{24}x_{119}\oplus x_{36}x_{145}\\
&\oplus x_{49}x_{72}x_{92}x_{119}
\end{split}
\end{align}
It is remarked in \cite{dubrova17} that  
\beq\label{eqfor255}
f_{255}(\xb)= f_{\text{Lin}}(\xb)\oplus f_{\text{NL}}(\xb)
\eeq
where $ f_{\text{Lin}}(\xb) $ is a linear polynomial component and $ f_{\text{NL}}(\xb) $ is a non-linear component. Further the feedback function at stage $ 217 $ is given precisely by
\[f_{217}(\xb)= x_{218}\oplus f_{\text{NL}}(\xb)_{(i-38)}\]
where the second term is given by the non-linear component $ f_{\text{NL}}(\xb) $ of $ f_{255} $, with every variable $ x_i $ changed to $ x_{i-38} $.\\

\noindent The feedback function of the Y-P Fibonacci LFSR can be re-written as 
\[ f_{L}(\xb)= (f_L(\xb)\oplus f_{\text{NL}}(\xb)) \oplus f_{\text{NL}}(\xb) \]
Using Dubrova's shift formula to transfer all the monomials in $ f_{\text{NL}}(\xb) $ from stage $ 255 $ to stage $ 217 $, we obtain the transformed feedback function at stage $ 217 $ as:
\[f_{217} = x_{218}\oplus f_{\text{NL}}(\xb)_{(i - 255 +217)} = x_{218}\oplus f_{\text{NL}}(\xb)_{(i-38)}\]
After the shift, the modified feedback function at stage $ 255 $ is \[f_{255}(\xb)= f_L(\xb)\oplus f_{\text{NL}}(\xb) \]
This is precisely the Galois NFSR $ F $ with exactly two non-linear feedback taps at stages $ 255 $ and $ 217 $.
\brmk
A general analysis and application of this `trick' of extracting and shifting an `invisible' feedback term by XORing twice is to be found in \cite{dubrova12}.
\ermk
\subsubsection*{Proposed Resolution of Espresso Equivalence Problem} As the NFSR $F$ is a uniform Galois NFSR, with terminal bit $ 217 $, the Y-P Fibonacci-to-uniform Galois algorithm will automatically generate a candidate for the non-linear output function. Hence, we can construct a non-linear filter generator cipher based on the Espresso alternate NFSR $ F $, which would be \emph{equivalent} to the Y-P LFSR ensemble in the sense discussed in \cite{yapa21}, \cite{yapa22}. Our proposed procedure is formally stated in Algorithm \ref{fekkinespes}. A theoretical validation and experimental results pertaining to this procedure are presented in the following subsections.
\begin{algorithm}
	\caption{\textsc{Generation of Espresso Alternate NFSR Cipher}}\label{fekkinespes}
	\smallskip
	
	\begin{algorithmic}[1]
		\STATE \textbf{Input:} Y-P $ n $-bit Fibonacci LFSR $ L $ and Non-linear function $ z_{\text{YP}}(\xb) $; Initial state $ \bfs $ of Fibonacci LFSR.\\
		\noindent \textbf{Output:} Espresso alternate NFSR $ F $ and Non-linear function $ z_F (\xb) $; Matching initial state $ \rb $ of NFSR $ F $.
		\STATE Use Dubrova's shift formula to generate NFSR $ F $:
		\beqn
		f_{255}\xrightarrow{f_{\text{NL}}(\xb)}f_{217}
		\eeqn
		\STATE Generate combined compensation list $ C $ for the shift.
		\STATE Obtain matching initial state for $ F $ by simple compensation using $ C $ (Dubrova's theorem/Y-P lemma).
		\STATE Obtain output function $ z_F (\xb) $ from $ z_{\text{YP}}(\xb) $ by iterative compensation in the descending order using $ C $.
	\end{algorithmic}
\end{algorithm}

\subsection{Validating the Solution to EEP}
In this subsection we present theoretical analysis in support of our proposed procedure for resolving the Espresso Equivalence problem.\\
As stated in Algorithm \ref{fekkinespes}, our procedure is based on using the Dubrova-Yao-Parampalli transformation framework to transform the Yao-Parampalli LFSR-based Espresso-equivalent design into a design based on the Espresso alternate NFSR $F$. It is evident from the analysis in Subsection \ref{subs:lofornofor}, that we have achieved the LFSR-NFSR correspondence through appropriate use of Dubrova's shift formula. The final, more critical, challenge is to ensure that the output of the non-linear function so obtained is identical for both the Y-P LFSR-based design and the NFSR $F$-based design derived by us. To that end, we first examine the recipe(s) proposed by Yao-Parampalli in \cite{yapa22}, to effect transformations between Fibonacci and uniform Galois NFSRs such that they preserve the output bitstreams of non-linear functions involving several stages of the NFSRs.
\subsubsection{Issues with the Yao-Parampalli Algorithms}
As discussed in Section \ref{subs:yap}, there are three proven results in \cite{yapa22}, Lemmas \ref{lemlyapa1} and \ref{lemlyapa2} along with the Theorem \ref{thm:lomyapa} (cf. p. 238, Section 3, \cite{yapa22}), all of which deal with the transformation of a Fibonacci to a uniform Galois NFSR. We have already discussed the lemmas, and established their equivalence with Dubrova's works (\cite{dubrova10}, \cite{dubrova14}), in computing matching initial states for the Fibonacci and Galois NFSRs. We now examine the following theorem of Yao-Parampalli which claims that their algorithm (cf. Section \ref{subs:yap}) achieves equivalence in the outputs of the respective non-linear functions.
\bthm[\textbf{Theorem 1}, \cite{yapa22}]\label{thm:lomyapa}
Let an $ n $-bit Fibonacci NLFSR [sic] transformed to a Type-II Galois NLFSR according to the Fibonacci-to-Type-II algorithm, then the internal states of the
NLFSR before and after transformation denoted by $ x_i^t $ and $ \hat{x}_i^t$ satisfy
\begin{align*}
\hat{x}_i^t & = x_i^t \oplus c_i(X^t),\,\,\text{for}\,\, k \in [\tau_1+1, n-1]\\
\hat{x}_i^t & = x_i^t,\,\,\text{for}\,\, k \in [0,\tau_1]
\end{align*}
and
\begin{align*}
x_i^t & = \hat{x}_i^t \oplus \overline{c_i}(\hat{X}^t),\,\,\text{for}\,\, k \in [\tau_1+1, n-1]\\
x_i^t & = \hat{x}_i^t,\,\,\text{for}\,\, k \in [0,\tau_1]
\end{align*}
and the two NLFSRs are equivalent.
\ethm
A comparison with Lemma \ref{lemlyapa1} and Lemma \ref{lemlyapa2} shows that the formulas in Theorem \ref{thm:lomyapa} have the same form with those presented in the lemmas, with single monomials in the former replaced by compensation entries in the latter. The main body of the proof of Theorem \ref{thm:lomyapa} consists of validating the stated formulas in simple instances. The crucial assertion in the theorem regarding ``equivalence'', which appears to be the real contribution of the Y-P formulation beyond Dubrova's framework, is handled in the concluding paragraph of the proof in \cite{yapa22} as follows.\\

\noindent``\emph{If the output function $ f_z $ takes $ x_i , i \in [\tau_1 +
	1, n - 1] $ as input, the corresponding tap in $ \overline{f_z} $ is compensated by $ C $. Since the compensation is done iteratively, the $ x_i $ in $ c_i $ is also compensated. Therefore, at each clock $ t $, the bit value of that input in $ \overline{f_z} $ is equal to $ \hat{x}_i^t \oplus \overline{c_i}(\hat{X}^t) $ which is equal to $ x_i^t $ according to (10) \emph{(cf. Theorem \ref{thm:lomyapa})}. Hence, the output function $ \overline{f_z} $ generate [sic] same bit as the original function $ f_z $ does at each clock.}'' \\
	
We point out the following issues with the above statement and the steps of transformation from ``Fibonacci-to-Type II'' (uniform Galois) on p. 238  and ``Type II-to-Fibonacci'' on p. 241 in \cite{yapa22}, respectively.
\begin{itemize}
	\slen{1.26}
	\item It appears from the quoted statement that the equivalence of the outputs of the non-linear functions before and after the transformation is achieved via the following steps:
	\begin{enumerate}
		\item Perform iterative compensation on the variables constituting the output function $f_z$ to obtain $\overline{f_z}$.
		\item Perform iterative compensation on the original initial state to obtain $\hat{x}_i^t \oplus \overline{c_i}(\hat{X}^t)$ (following the notation in Lemma \ref{lemlyapa2}).
		\item Evaluate the iteratively compensated output function $\overline{f_z}$ on the iteratively compensated initial state.
	\end{enumerate}
\textbf{Issue:} It is simply stated, without proof, that this evaluation is identical with the original output function evaluated on the original initial state.
\item As shown in Subsection \ref{subs:yap}, the steps for ``Fibonacci-to-Type II'' transformation in \cite{yapa22} consist of the following sequence for producing identical function outputs: iterative compensation of output function followed by simple compensation of the initial state to be used as argument.\\
\textbf{Issue:} This is at variance with the suggested procedure in the ``proof'' of Theorem \ref{thm:lomyapa}.
\item The stated procedure for ``Type II-to-Fibonacci'' in \cite{yapa22} transformation suggests  simple compensation of output function followed by iterative compensation of the initial state to be used as argument.\\
\textbf{Issue:} In this case, swapping the sequence of compensation operation that was suggested for the earlier case of ``Fibonacci-to-Type II'' adds to the ambiguity. More so, as no reason or validation is provided for either. 
\end{itemize}
\brmk
It is noted in this context that the proof of Theorem 1 as well as that of Lemma 2 in \cite{yapa22} refers to a ``Lemma 3'' for some sort of validation. Unfortunately, the said result does not exist in the entire paper, even though it is referred to elsewhere as well.
\ermk
\brmk
In \cite{yapa22}, the authors present an illustration of their algorithm for Fibonacci-to-uniform Galois (Type-II) transformation (cf. Algorithm~2, Section 3, p. 240). They begin with a 4-bit Fibonacci NFSR with initial state $X^0=\{x_{0}^{0},x_{1}^{0},x_{2}^{0},x_{3}^{0}\}=\{1,1,0,1\}$. However, in the later steps of the algorithm, specifically from Step~5, the initial state is specified as $\{0,1,1,1\}$. No reason is given for this modification, whatsoever. In the following steps, simple compensation of the modified state $\{0,1,1,1\}$ is used to generate $\{0,1,0,0\}$, rather than the declared initial state.
\ermk
\subsubsection*{Resolving the Issues and Proof of Concept} In what follows we resolve the above issues in the Yao-Parampalli algorithms to establish the correct sequence of operations of simple and iterative transformations. Further, we formally prove that our prescribed sequence indeed achieves the following objective of Yao-Parampalli's algorithms:\\

\noindent\emph{The modified non-linear function obtained after transformation, with a modified input state as the argument, will output exactly the same bitstream as the original function with a \emph{matched} input state as argument.}\\

\noindent We begin with a result on the composition of the dependence lists of the compensation entries corresponding to Fibonacci-to-uniform Galois NFSR transformations. This is an ingredient in the proof of the main theorem which will provide theoretical validation to the claimed equivalence.
\blem\label{lemcompindu}
For the D-Y-P equivalent transformation from an $n$-bit Fibonacci to a uniform Galois NFSR, with terminal bit $0\leq\tau< n-1$, the entries $c_i$ of the compensation list $C$ satisfy the following condition:
\[\max \text{dep}(c_i) \leq i-1, \, \forall \, i \in \{\tau+1, \tau +2, \ldots,n-1\}\]
\elem
\begin{proof}
	Consider the following shift from the feedback function at the $(n-1)$-th stage of the Fibonacci NFSR:
	\[f_{n-1}\xrightarrow{g^\ast_{\tau}} f_{\tau}\]
	where we use the notation of Theorem \ref{thm:chorcompenis}, with the resulting shifted function at $f_{\tau}$ denoted by $g_{\tau}$. As the transformed Galois NFSR, defined using the notation of Theorem \ref{thm:compendubri}, is uniform, we have:
	\[\text{dep} (g_i)\subseteq [0,\tau], \,\, \forall \, i \in \{\tau+1, \tau +2, \ldots,n-1\}\]
	At some such $i$, it follows from Theorem \ref{thm:compendubri} and Theorem \ref{thm:chorcompenis} that the compensation expression is defined by the terms:
	\[g_{i-1}, g_{i-2}|_{+1}, \ldots, g_{\tau}|_{+(i-\tau-1)}\]
	The first possibly non-zero compensation term is $c_{\tau+1}$, which is simply given by $g_\tau$. Hence, it follows that, for $i= \tau+1$,
	\[\max \text{dep}(c_i) = \max \text{dep}(c_{\tau+1}) = \tau \leq i-1\]
	For $i=\tau +2$, we have the contributing terms $g_{\tau+1}$ and $g_\tau|_{+1}$ where
	\[\text{dep}(g_{\tau +1})\subseteq [0,\tau], \, \, \text{while} \,\, \text{dep}(g_\tau|_{+1})\subseteq [1,\tau+1]\]
	Hence, the assertion holds for $i=\tau +2$ as well, since $\max \text{dep}(c_{\tau+2}) \leq \tau +1 $.\\
	Further, it is readily observed that for any $i$, $\max \text{dep}(c_i) $ is determined by the term involving $g_{\tau}$, and so, we have:
	\[\max \text{dep}(c_i)\leq \text{dep}(g_{\tau}|_{+(i-\tau-1)})\leq (\tau + i-\tau -1)=i-1.\]
	This proves the assertion.
\end{proof}
\qed
We now present the main theorem of this subsection, which explicitly confirms the design of the transformed non-linear output function as well as the choice of input which achieves equivalence of the output bitstreams before and after a Fibonacci-uniform Galois transformation. Some new notation is introduced for this purpose as follows.\\
Let a function $ F $ in, say, $ n $ Boolean variables $ x_0, \ldots, x_{n-1} $, be denoted as: $ F(X) = F(x_0, \ldots, x_{n-1})$. If these variables are iteratively compensated (in descending order of the indices) by some array $ C= [c_0, c_1, \ldots, c_{n-1}] $, where each $ c_i $ is a Boolean function in $ X $, the form obtained from $ F $ is termed the iteratively compensated $ F $. We denote the iteratively compensated function with the superscript $ \iota $ as $ F^{\iota} $. Similarly, if the variables are compensated by direct XORing with the entries in $C$ - what has been termed simple compensation - the compensated function form has the superscript $\sigma$: $F^{\sigma}$.\\
If the function $F$ is evaluated at some modified value of $X$, where each $x_i$ component has been iteratively compensated as described above, the evaluation is denoted as $F(X^{\iota})$, with the analogous form for simple compensation of $X$ denoted by $F(X^{\sigma})$.\\
A combination of these compensation operations, where both the function and the argument have been compensated, will carry both the superscripts. For instance, an iteratively compensated function $F$ evaluated on simple compensated argument $X$ will be denoted by $F^{\iota}(X^{\sigma})$.\\
Using the above notation, we now state a result which settles the question of the correct choice of the compensation operations for the output function and initial state to achieve matching output bitstreams for Fibonacci to uniform Galois NFSR transformation. 

\bthm\label{thm:itorFiboGalois} For Fibonacci-to-Uniform Galois NFSR transformations according to the D-Y-P framework, any choice of compensation for the output function and the initial state will produce matching output bitstreams before and after transformation. In particular, denoting the Fibonacci output function as $f$ and the Fibonacci initial state as $X$, we have, using the notation defined above, the following equalities:
\[f^{\iota}(X^{\sigma})= f^{\sigma}(X^{\iota})= f^{\iota}(X^{\iota})= f^{\sigma}(X^{\sigma})= f(X)\]  
\ethm
\begin{proof}
	We prove the theorem in two parts as follows:
	\begin{enumerate}
		\item $f^{\iota}(X^{\sigma})= f^{\sigma}(X^{\iota})= f(X)$;
		\item $ f^{\iota}(X^{\iota})= f(X)$ and $f^{\sigma}(X^{\sigma})= f(X)$.
	\end{enumerate}
	
\noindent \textbf{Part 1.:} Consider $X_1=(x_0, x_1)$, where the subscript corresponds to the maximum index in the list of variables. By Lemma \ref{lemcompindu}, the corresponding entries of the compensation list are specified as follows: $c_0$ is a constant and $c_1 = c_1(x_0)$. Therefore, we have:
\begin{align*}
	f^{\sigma}(X_1) &= f(x_0 \oplus c_0,x_1\oplus c_1(x_0))\\
	f^{\sigma}(X_1^\iota) &= f((x_0 \oplus c_0)\oplus c_0,(x_1\oplus c_1(x_0 \oplus c_0))\oplus c_1(x_0 \oplus c_0))\\
	& = f(x_0, x_1).
\end{align*}
Similarly,
\begin{align*}
	f^{\iota}(X_1) &= f(x_0 \oplus c_0,x_1\oplus c_1(x_0 \oplus c_0))\\
	f^{\iota}(X_1^\sigma) &= f((x_0 \oplus c_0)\oplus c_0,(x_1\oplus c_1(x_0 ))\oplus c_1((x_0 \oplus c_0) \oplus c_0))\\
	& = f(x_0, x_1\oplus c_1(x_0 )\oplus c_1(x_0 ))\\
	& = f(x_0, x_1).
\end{align*}
Hence, we have: $ f^{\sigma}(X_1^\iota) = f^{\iota}(X_1^\sigma)= f(X_1).$
We prove the result by induction on $i$, where $X_i = (x_0,x_1, \ldots, x_{i})$.\\
Assume, as induction hypothesis, that the result holds for all $i\leq k-1$. We will prove the result holds for $i =k$, i.e.
\beq\label{eqkokkitsim}
f^{\sigma}(X_k^\iota) = f^{\iota}(X_k^\sigma)= f(X_k).
\eeq
The form of the function is compensated first and then the argument is compensated prior to evaluation. Hence, it can be verified that 
\[f^\iota (X_{k-1}^\sigma) = f((x_0^\iota)^\sigma, \ldots, (x_{k-1}^\iota)^\sigma)\]
with an analogous expansion for $f^\sigma(X_{k-1}^\iota)$.\\
Either way, the result in Equation \ref{eqkokkitsim} is proved if the following holds, making use of the induction hypothesis:
\[(x_k^\sigma)^\iota = (x_k^\iota)^\sigma = x_k \,\,\text{given}\,\, (X_i^\sigma)^\iota = (X_i^\iota)^\sigma = X_i,\,\,\forall \, i\leq k-1.\]
We have:
\begin{align*}
	(x_k^\sigma)^\iota & = (x_k \oplus c_k(X_{k-1}))^{\iota}\\
	&= x_k^\iota \oplus c_k(X_{k-1}^\iota)\\
	&= x_k \oplus c_k(X_{k-1}^\iota)\oplus c_k(X_{k-1}^\iota)\\
	&= x_k.
\end{align*}
On the other hand,
\begin{align*}
	(x_k^\iota)^\sigma & = (x_k \oplus c_k(X_{k-1}^{\iota}))^\sigma\\
	& = (x_k \oplus c_k(X_{k-1}))\oplus c_k((X_{k-1}^\iota)^\sigma)\\
	& = x_k \oplus c_k(X_{k-1})\oplus c_k(X_{k-1})\,\, \, (\text{by induction})\\
	& = x_k.
\end{align*}
Hence the equalities in Equation \ref{eqkokkitsim} hold.\\

\noindent \textbf{Part 2.:} Consider evaluation of $f^\iota(X_1^\iota)$ in two stages as follows. First we compute the form of the iteratively compensated function (in descending order):
\[f^\iota(X_1) = f(x_0^\iota, x_1^\iota)= f(x_0 \oplus c_0, x_1 \oplus c_1(x_0 \oplus c_0))\]
Next we evaluate the above compensated function on an iteratively compensated argument, i.e. replacing all the variables in the above expression with iteratively compensated variable forms \emph{in descending order}. Therefore, we have:

\begin{align*}
	f^\iota(X_1)& \xrightarrow{x_1\rightarrow x_1^\iota} f(x_0 \oplus c_0, (x_1 \oplus c_1(x_0 \oplus c_0))\oplus c_1(x_0 \oplus c_0))= f(x_0 \oplus c_0, x_1)\\
	& \xrightarrow{x_0\rightarrow x_0^\iota} f((x_0 \oplus c_0 )\oplus c_0, x_1)= f(x_0,x_1)\\
	& =f^\iota(X_1^\iota).
\end{align*}
In the same manner as computed above for $k=1$, the first pass involving the iterative compensation of the argument in $f^\iota(X_k)$, for any positive integer $k\leq n-1$, will be:
\[f^\iota (X_k)\xrightarrow{x_k\rightarrow x_k^\iota} f(X_{k-1}^\iota,(x_k \oplus c_k(X_{k-1}^\iota))\oplus c_k(X_{k-1}^\iota))= f(X_{k-1}^\iota,x_k).\]
Therefore, the successive passes for the iterative compensation of the argument will clearly return
\[f^\iota(X_k^\iota) = f(X_k).\] 
Finally, we have:
\begin{align*}
	f^\sigma(X_1^\sigma) &= f(x_0^\sigma\oplus c_0,x_1^\sigma\oplus c_1)\\
	&= f((x_0\oplus c_0)\oplus c_0, (x_1\oplus c_1)\oplus c_1)\\
	&= f(x_0,x_1)
\end{align*}
which is readily extended from $k=1$ to any positive integer $k\leq n-1$:
\[f^\sigma(X_k^\sigma) = f(X_k).\]
Hence the theorem.
\end{proof}
\qed
In view of the results in Theorem \ref{thm:itorFiboGalois}, we now have sufficient clarity as regards the D-Y-P Fibonacci-to-uniform Galois NFSR transformation for choosing either iterative or simple compensation for the original output function as well as the original initial state to produce matched output bitstreams post transformation.\\

\noindent Our experimental results in the next subsection further corroborate the statement of Theorem \ref{thm:itorFiboGalois}, which removes all ambiguity regarding the steps of this particular NFSR transformation.
\subsection{Experimental Results}
In \cite{yapa21}, a $256$-bit Fibonacci LFSR was derived from the Espresso design NFSR $G$ using (Dubrova-Yao-Parampalli) uniform Galois-to-Fibonacci transformation. The transformation generated a counterpart to the non-linear output function defined for Espresso which had $2289$ monomials in $104$ variables and degree $12$. It was claimed that this LFSR/non-linear function ensemble, following the recipe of simple compensation of the output function and iterative compensation of the initial state (\cite{yapa21}, \cite{yapa22}), generated identical bitstreams with Espresso. We have experimentally verified this claim, which is depicted in the left halves of the schematics in Figure \ref{fig:akbebhul} and Figure \ref{fig:reseep}.\\ 

\noindent Building on the verified claim in \cite{yapa21}, we have successfully used the D-Y-P Fibonacci-to-uniform Galois transformation on the LFSR-based ensemble to generate a non-linear output function corresponding to the Espresso alternate Galois NFSR $F$ (cf. Algorithm \ref{fekkinespes}). Our non-linear function consists of $3656$ monomials in $88$ variables and has algebraic degree $18$. This NFSR $F$/non-linear function ensemble generates identical output bitstream, using iterative compensation of the output function and simple compensation of the initial state (cf. \cite{yapa22}), with the Fibonacci LFSR/non-linear function ensemble in \cite{yapa21}.\\
\emph{Further, in accordance with the results in Theorem \ref{thm:itorFiboGalois}, we have verified that the output bitstream of our NFSR $F$-based ensemble matches the Yao-Parampalli LFSR-based ensemble using \emph{both simple and iterative compensation} separately on their output function and the LFSR initial state.}\\

\noindent We have, therefore, positively resolved the Espresso Equivalence Problem as stated in Section \ref{secyapdubri} (cf. Figure \ref{fig:dubrigupi}), by generating a non-linear output function corresponding to the Espresso alternate NFSR $F$ which `inherits' the equivalence with the original Espresso design through the following steps:
\begin{enumerate}
	\item Espresso Galois NFSR $G$ ensemble is output-equivalent to Yao-Parampalli Fibonacci LFSR ensemble
	\item Yao-Parampalli Fibonacci LFSR ensemble is output-equivalent to Espresso NFSR $F$ ensemble (this paper).
	\item Therefore, Espresso Galois NFSR $G$ ensemble is output-equivalent to Espresso NFSR $F$ ensemble.
\end{enumerate}
This is graphically depicted in Figure \ref{fig:reseep}.

\begin{figure}[H] 
	\centering
	\includegraphics[width= \textwidth]{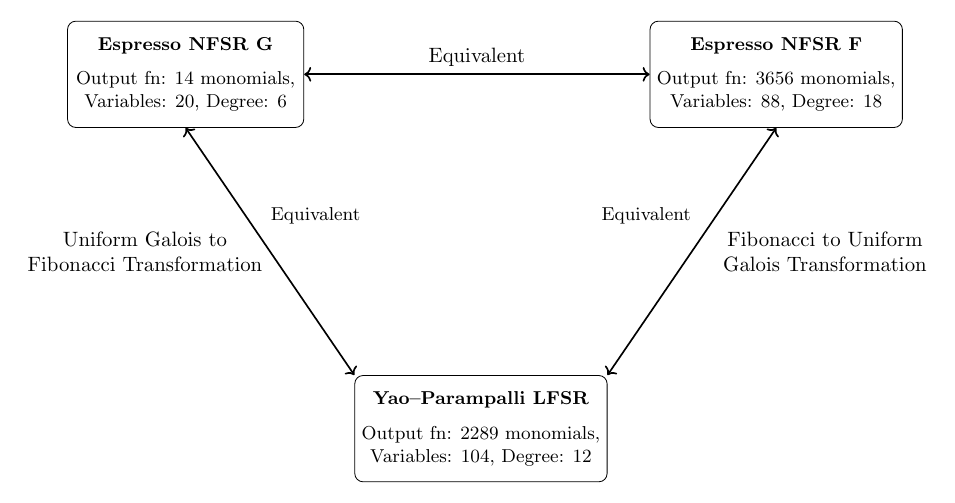}		
	\caption{Resolving the Espresso Equivalence Problem}
	\label{fig:reseep}
\end{figure}

\section*{Conclusion and Future Work}
The design methodology of the stream cipher Espresso is theoretically significant as well as of practical interest as it attempts to harness the benefits of a Galois NFSR-based implementation while promising simpler analysis through an equivalent design. But the validity of its security claims against several standard cryptanalytic techniques rested on the equivalence of two designs discussed in \cite{dubrova17}. We have termed this the Espresso Equivalence Problem and have demonstrated the first instance for its positive resolution.\\ 
In order to demonstrate the said equivalence, we have attempted a partial unification of E. Dubrova and Yao-Parampalli's transformation techniques. The insights obtained from the above analysis led us to recover the Espresso alternate NFSR $F$ from the Y-P LFSR, and, finally, a matching non-linear function for $ F $. This has opened up the possibility of revisiting the security of Espresso through a proper analysis of the now truly output-equivalent alternate NFSR $F$ and non-linear function ensemble. However, it remains an open question whether the two output-equivalent versions will provide any significant advantage in analysis, as evident from the following list of properties:
\begin{itemize}
	\item \textbf{Espresso:} \cite{dubrova17} Galois NFSR $G$ with $14$ non-linear taps; Output function: $14$ monomials, $20$ variables, algebraic degree $6$.
	
	\item \textbf{Y-P Design:} \cite{yapa21} Fibonacci LFSR; Output function: $2289$ monomials, $104$ variables, algebraic degree $12$.
	
	\item \textbf{Espresso Equivalent:} (this paper) Galois NFSR $F$ with $2$ non-linear taps; Output function: $3656$ monomials, $88$ variables, algebraic degree $18$.	
\end{itemize}
In course of establishing the Espresso equivalence, we have demonstrated that (cf. Theorem \ref{thm:itorFiboGalois}), for the Fibonacci-to-uniform Galois transformation, either compensation procedure applied to the output function-initial state would suffice to match the output. This assertion has also been verified by our experiments. However, this begs the obvious question:\\
Why is the sequence of compensation operations applied in \cite{yapa22}, with respect to Fibonacci-to-uniform Galois transformation, to the output function and the initial state, \emph{interchanged} for the uniform Galois-to-Fibonacci case?\\
The answer to the above question may lie in formulating and proving analogous results of Lemma \ref{lemlyapa1}, Lemma \ref{lemlyapa2}, and Theorem \ref{thm:lomyapa} for the case of uniform Galois-to-Fibonacci case. Further, in addition to the Fibonacci-uniform Galois transformation, three other NFSR designs have been discussed in \cite{yapa22}. However, as stated earlier, rigorous proof of concept has only been provided for Fibonacci-to-uniform Galois case. Proving (or disproving) the (output) equivalence in all these cases would be a worthwhile research direction.\\
Recently, there is a growing body of research on NFSR equivalence, for instance, \cite{zhaoequiv20}, \cite{zhongTrivnekk21} and \cite{panequivGaloisFibo}. In \cite{zhongTrivnekk21}, the authors claim to establish that the $288$-stage Galois NFSR of Trivium cannot be shown to be equivalent to a $288$-stage Fibonacci NFSR. We hope that further investigation would lead to the creation of a general template for equivalence-preserving shift register transformations. Such a template, if rigorously formulated, would open up the exciting possibility of obtaining equivalent circuits for existing designs of proven security like Grain and Trivium, and possibly identify variants with optimized implementation.
\subsubsection*{Acknowledgement}
%The authors would like to thank the anonymous reviewers for their detailed comments that improved both the content and presentation. 
The assistance of Debendranath Das with initial experimental verifications and the inputs from Subhamoy Maitra, Mridul Nandi and Sabyasachi Karati are gratefully acknowledged. The work of Anirban Ghatak is supported by SAG, DRDO, Delhi, and R. C. Bose Centre for Cryptology and Security, ISI, Kolkata.
\bibliography{espes}

\end{document}